\documentclass{lmcs}
\pdfoutput=1

\usepackage{lastpage}
\renewcommand{\dOi}{14(4:10)2018}
\lmcsheading{}{1--\pageref{LastPage}}{}{}%
{Nov.~27,~2017}{Oct.~30,~2018}{}

\usepackage[utf8]{inputenc}

\usepackage{graphicx}
\usepackage{epic}
\usepackage{eepic}
\usepackage{url}
\usepackage{xspace}
\usepackage{amsmath}
\usepackage{amsthm}

\usepackage{marvosym}

\usepackage{amssymb}
\usepackage{stmaryrd}
\usepackage{tikz}
 \usetikzlibrary{trees}
 \usetikzlibrary{shapes}
 \usetikzlibrary{arrows}
 \usetikzlibrary{matrix}
\usetikzlibrary{fit}
\usetikzlibrary{backgrounds}	

 \usepackage{algorithm}
 \usepackage{algpseudocode}
 
\usepackage[square,numbers]{natbib}
\newcommand{\PP}{\mathrm{P}}
\newcommand{\NP}{\mathrm{NP}}
\newcommand{\coNP}{\mathrm{coNP}}

\newcommand{\co}{\mathrm{co}\mbox{-}}
\newcommand{\DELTA}[1]{\Delta_{#1}\mathrm{P}}
\newcommand{\THETA}[1]{\Theta_{#1}\mathrm{P}}
\newcommand{\DELTAlog}[1]{\Delta_{#1}\mathrm{P}[O(\log n)]}
\newcommand{\SIGMA}[1]{\Sigma_{#1}\mathrm{P}}
\newcommand{\PI}[1]{\Pi_{#1}\mathrm{P}}
\newcommand{\OPT}{\mathrm{OptP}}
\newcommand{\OPTlog}{\mathrm{OptP}[O(\log n)]}

\renewcommand{\phi}{\varphi}
\renewcommand{\theta}{\vartheta}

\newcommand{\cal}[1]{\mathcal{#1}}

\newcommand{\nop}[1]{}

\newcommand{\var}{\operatorname{var}}

\newcommand{\npmax}{{\sc NP-Max}\xspace}
\newcommand{\lognpmax}{{\sc LogNP-Max}\xspace}
\newcommand{\nsat}{N_{\rm SAT}}

\newcommand{\lexmax}{{\sc LexMaxSat}\xspace}
\newcommand{\loglexmax}{{\sc LogLexMaxSat}\xspace}

\newcommand{\cardmax}{{\sc CardMaxSat}\xspace}
\newcommand{\cardmin}{{\sc CardMinSat}\xspace}

\newcommand{\indsetWithBound}{{\sc MaxIndependentSet}\xspace}

\newcommand{\revdal}{\circ_{D}}
\newcommand{\revsat}{\circ_{S}}
\newcommand{\revR}{\circ_{r}}

\newcommand{\mmod}{\mathrm{mod}}

\newcommand{\implication}{{\sc BR}-{\sc Implication}\xspace}
\newcommand{\mc}{{\sc BR}-{\sc Model Checking}\xspace}

\newcommand{\Var}{\mathit{Var}}

\newcommand{\relevance}{$\preceq$-{\sc Relevance}\xspace}
\newcommand{\relevanceplain}{{\sc Relevance}\xspace}

\newcommand{\pap}{\mathcal{P}}
\newcommand{\tuple}[1]{\langle #1 \rangle}

\newcommand{\s}{\mathcal{S}}

\newcommand{\clonename}[1]{\mathrm{#1}}

\newcommand{\frclosure}[1]{\ensuremath{\langle#1\rangle_{fr}}}

\newcommand{\mcal}[1]{\mathcal{#1}}

\newcommand{\relOr}[1]{\ensuremath{\text{Or}_{#1}}}

\newcommand{\cardminsat}{\text{\sc CardMinSat}}

\newcommand{\sat}{\text{\sc Sat}}

\newcommand{\coclone}[1]{\ensuremath{\langle#1\rangle}}
\newcommand{\cocloneFr}[1]{\ensuremath{\langle#1\rangle_{\text{fr}}}}
\newcommand{\cocloneFrnoeq}[1]{\ensuremath{\langle#1\rangle_{\text{fr},\ne}}}
\newcommand{\coclonenoexistnoeq}[1]{\ensuremath{\langle#1\rangle_{\not\exists,\ne}}}
\newcommand{\inv}[1]{\mathrm{Inv}\!\left(#1\right)}
\newcommand{\pol}[1]{\mathrm{Pol}\!\left(#1\right)}

\newcommand{\set}[1]{\ensuremath\left\{#1\right\}}

\newcommand{\daffine}{width-$2$ affine\xspace}

\newcommand{\ie}{\text{i.e.}\xspace} 
\newcommand{\eg}{\text{e.g.}\xspace}

\newlength\problemlength
\newcommand\problem[3]{%
\begin{list}{}{\labelwidth\problemlength \labelsep.7em \rightmargin1.5em
\leftmargin\problemlength \advance\leftmargin by3em
\parsep0ex \itemsep.2ex plus.1ex}
\item[{\sl Problem:\hfill}] #1 \item[{\sl Input:  \hfill}] #2
\item[{\sl Output: \hfill}] #3
\end{list}
}
\newcommand\dproblem[3]{%
\begin{list}{}{\labelwidth\problemlength \labelsep.7em \rightmargin1.5em
\leftmargin\problemlength \advance\leftmargin by3em
\parsep0ex \itemsep.2ex plus.1ex}
\item[{\sl Problem:\hfill}] #1 \item[{\sl Input:  \hfill}] #2
\item[{\sl Question: \hfill}] #3
\end{list}
}

\title{Do Hard SAT-Related Reasoning Tasks Become Easier in the Krom Fragment?}

\author{Nadia Creignou\rsuper{1}}
\address{\lsuper{1}Aix Marseille Univ, Universit\'{e} de Toulon, CNRS, LIS, Marseille, France}
\email{nadia.creignou@univ-amu.fr}
\author{Reinhard Pichler\rsuper{2}}
\address{\lsuper{2}TU Wien, Favoritenstra\ss e 9--11, A-1040 Vienna, Austria}
\email{pichler@dbai.tuwien.ac.at}
\author{Stefan Woltran\rsuper{2}}
\address{\vskip-7pt}
\email{woltran@dbai.tuwien.ac.at}

\begin{document}

\begin{abstract}
Many reasoning problems are based on the problem of satisfiability (SAT). While SAT itself becomes easy when restricting the structure of the formulas in a certain way, the situation is more opaque for more involved decision problems.  
We consider here the \cardmin problem which asks, given a propositional formula $\phi$ and an atom $x$, whether $x$ is true in some cardinality-minimal model of $\phi$. This problem is easy for the Horn fragment, but, as we will show in this paper, remains  $\THETA 2$-complete (and thus $\NP$-hard) for the Krom fragment (which is given by formulas in CNF where clauses have at most two literals).  We will make use of this fact to study the complexity of reasoning tasks in belief revision and logic-based abduction and show that, while in some cases the restriction to Krom formulas  leads to a decrease of complexity, in others it does not.  We thus also consider the \cardmin problem with respect to additional restrictions to Krom formulas towards a better understanding of the tractability frontier of such problems.
\smallskip

\end{abstract}

\maketitle

\section{Introduction}
\label{sect:introduction}



By Schaefer's famous theorem \cite{stoc/Schaefer78}, we know that the SAT problem becomes tractable under certain syntactic restrictions such as the restriction to Horn formulas
(i.e., formulas in CNF 
where clauses have at most one positive literal) or to 
Krom formulas (i.e., clauses have at most two literals).
Propositional formulas play an important role in reasoning problems in a great variety of areas
such as belief change, logic-based abduction, closed-world reasoning, etc. 
Most of the relevant problems in these areas are intractable.
It is therefore a natural question whether restrictions on the 
formulas involved 
help to decrease the high 
complexity. 
While
the restriction to Horn formulas is usually well studied, other
restrictions in Schaefer's framework have received much less attention. 
In this work, we
have a closer look at the restriction to Krom formulas and
its effect on the complexity of several hard 
reasoning problems from the AI domain.

%
%


One particular source of complexity of such problems is the involvement of some notion
of {\em minimality\/}.
In closed world reasoning,
we add negative literals or more general formulas to a theory only if they are true
in every minimal model of the theory in order to circumvent inconsistency.
In belief revision we want to 
revise a given belief set  by some new belief. To this end, we retain only those models of the new belief which have {\em minimal\/} distance to the models of the given belief set. 
In abduction, we search for a 
subset of the hypotheses (i.e., an ``explanation'') that is consistent with the given theory and which -- together with the theory -- explains (i.e., logically entails) all manifestations. Again, one is usually not content with any subset 
of the hypotheses but with a minimal one. 
%
%
%
However, 
different notions of minimality might be considered, 
in particular, minimality w.r.t.\ set inclusion or w.r.t.\ cardinality.
In abduction, these two notions are directly applied to explanations 
\cite{jacm/EiterG95}. 
In belief revision, 
one is interested in the models of the new belief which have  minimal distance from
the models of the given belief set. 
Distance between models is defined here via 
the symmetric set difference $\Delta$
of the atoms assigned to true in the compared models.
Dalal's revision operator \cite{Dalal88} seeks to minimize the cardinality of $\Delta$ while
Satoh's 
operator \cite{fgcs/Satoh88} defines the minimality of $\Delta$
in terms of set inclusion.

%
%

The chosen notion of minimality usually has a significant impact on the complexity of the 
resulting reasoning tasks. If minimality is defined in terms of cardinality, 
we often get problems that are complete for some class $\DELTAlog k$ with $k \geq 2$
(which is also referred to as $\THETA k$, see \cite{icalp/Wagner88}). For instance, 
two of the most common reasoning tasks in belief revision (i.e., model checking and implication)
are $\THETA 2$-complete for Dalal's revision operator
\cite{ai/EiterG92,jcss/LiberatoreS01}.
Abduction aiming at cardinality-minimal 
explanations is $\THETA 3$-complete
\cite{jacm/EiterG95}.
If minimality is defined in terms of set inclusion, 
completeness results for one of the classes $\SIGMA k$ or
$\PI k$ for some $k \geq 2$ are more common. For instance, belief revision with Satoh's
revision operator becomes $\SIGMA 2$-complete (for model checking) 
\cite{jcss/LiberatoreS01}
respectively $\PI 2$-complete
(for implication)
\cite{ai/EiterG92}. Similarly, abduction is $\SIGMA 2$-complete if we test 
whether some hypothesis is contained in a subset-minimal explanation
\cite{jacm/EiterG95}.

%
%
%


For the  above mentioned problems, various ways to decrease
the complexity have been studied.
Indeed, 
in almost all of these cases, 
a restriction of the involved formulas to Horn 
makes the complexity drop by one level in the polynomial hierarchy.
Only belief revision with Dalal's revision operator remains $\THETA 2$-complete in the Horn case,
see \cite{ai/EiterG92,jcss/LiberatoreS01,jacm/EiterG95}.
The restriction to Krom has not been considered yet for these problems. In this paper we 
show that the picture is very similar to the Horn case. Indeed, for the considered problems
in belief revision and abduction, we get for the Krom case exactly the same complexity classifications.
Actually, we choose the problem reductions for our hardness proofs in such a way that we thus
also strengthen the previous hardness results by showing that they even hold if formulas are restricted to Horn {\em and\/} Krom at the same time. 

%
%


To get a deeper understanding why 
certain problems
remain $\THETA 2$-hard in the Krom case, we have a closer look at a related variant of
the SAT problem, where also cardinality minimality is involved. We define the cardinality of a truth assignment (or a model) as   the cardinality of the set of variables assigned true. 
Thus, we consider 
 the \cardmin problem: 
given a propositional formula $\phi$ and an atom $x$, is 
$x$ true in some {\em cardinality-minimal\/} model
of $\phi$? It is easy to show that this  problem is  $\THETA 2$-complete. 
If $\phi$ is restricted to Horn, then this problem becomes trivial, 
since Horn formulas have a unique minimal model which can be efficiently computed.
But what happens if 
we restrict $\phi$ to Krom? We 
show that 
$\THETA 2$-completeness holds also in this case. 
This hardness result will then also be very convenient for proving our $\THETA
2$-hardness results for the considered problems in belief revision and
abduction. Since  \cardmin 
seems to be central in evaluating the
complexity of reasoning problems in which some kind of cardinality-minimality is involved
we investigate its complexity in a deeper way, 
in particular by
characterizing the tractable cases.

\smallskip

\noindent
The {\bf main contributions of the paper} are as follows:

\begin{itemize}
\item 
{\em Prototypical problems for $\DELTA 2$ and $\THETA 2$.}
In Section~\ref{sect:delta2}, we first review 
the prototypical $\DELTA 2$-problem
\lexmax.
Hardness for this problem is due to
\cite{jcss/Krentel88} but only follows implicitly from that work.
We reformulate the proofs in terms of standard terminology 
in order to explicitly prove $\THETA 2$-completeness for the related problem
\loglexmax in an analogous way. Note that \loglexmax, which will be the basis for our forthcoming results, is
not mentioned explicitly in \cite{jcss/Krentel88}.


\item 
{\em SAT variants.}
In Section~\ref{sect:theta2}, we
investigate the complexity of the \cardmin problem and 
the analogously defined \cardmax problem.
Our central result is that  these problems remain
$\THETA 2$-complete
for formulas in Krom form.

\item 
{\em Applications.} We investigate several reasoning problems in the 
areas of belief revision 
and  abduction
in Section \ref{sect:belief-revision}.
For the restriction to Krom form, we 
establish the same complexity classifications as for the previously known restriction to Horn form. 
In fact, we thus also strengthen the previously known hardness results by showing
that hardness even holds for the simultaneous restriction to Horn {\em and\/} Krom.

\item {\em Classification of \cardmin inside the Krom fragment.} In
Section \ref{sect:classification} we investigate the complexity of \cardmin
\emph{within} the Krom fragment in order to identify 
the necessary additional syntactic
restrictions 
towards tractability. To this aim we 
use the well-known framework by Schaefer and obtain a complete complexity
classification of the problem.
\end{itemize}

\section{Preliminaries}
\label{sect:preliminaries}

{\bf Propositional logic.} 
We assume familiarity with the basics of propositional logic
\cite{BenAri03}. We only want to fix some notation and conventions here. We say that a formula is in $k$-CNF (resp.\ $k$-DNF) for $k \geq 2$ if all its clauses (resp.\ terms)
have at most $k$ literals. 
Formulas in $2$-CNF are also called \emph{Krom}. 
A formula is called \emph{Horn} (resp.\ \emph{dual Horn}) if it is in CNF and in each clause at most one literal is positive (resp.\ negative). 
It is convenient to identify truth assignments with the set of variables which are true in an assignment. It is thus possible to consider the cardinality and the subset-relation on the models of a formula. If in addition an order is defined on the variables, we may alternatively identify truth assignments with bit vectors, where we encode true (resp.\ false) by 1 (resp.\ 0) and arrange the propositional variables in decreasing order. We thus naturally get a lexicographical order on truth assignments.

\smallskip
\noindent
{\bf Complexity Classes.}
All complexity results in this paper refer to classes in the Polynomial Hierarchy (PH) \cite{Papadimitriou94}.
The building blocks of PH are the classes  $\PP$ and $\NP$ of decision problems solvable in 
deterministic resp.\ non-deterministic polynomial time. 
The classes 
$\DELTA k$, $\SIGMA k$, and $\PI k$ of PH are inductively defined as
$\DELTA 0 = \SIGMA 0 = \PI 0 = \PP$ and 
$\DELTA{k+1} = \PP^{\SIGMA k}$,
$\SIGMA{k+1} = \NP^{\SIGMA k}$, 
$\PI{k+1} = \co\SIGMA{k+1}$, where
we write $\PP^{\cal C}$ (resp.\ 
$\NP^{\cal C}$) for the class of decision problems that can be decided by a deterministic (resp.\ non-deterministic) Turing machine  in polynomial time using an oracle for the class ${\cal C}$. 
The prototypical $\SIGMA k$-complete problem is 
$\exists$-QSAT$_k$, (i.e., quantified satisfiability with $k$ alternating blocks of quantifiers, starting with $\exists$), where we have to decide the satisfiability of  a formula $\phi = \exists X_1 \forall X_2 \exists X_3 \dots Q X_k \psi$ 
(with $Q = \smash{\exists}$ for odd $k$ and  $Q = \smash{\forall}$ 
for even $k$) with no free propositional variables. 
W.l.o.g., one may assume here that $\psi$ is in 3-DNF (if $Q = \smash{\forall}$), resp.\
in 3-CNF (if $Q = \smash{\exists}$).
In \cite{icalp/Wagner88}, several restrictions on the oracle calls 
in a $\DELTA k$ computation have been studied. If on input $x$ with
$|x| = n$ at most $O(\log n)$ calls 
to the $\SIGMA{k-1}$ oracles are allowed, then we get the class 
$\PP^{\SIGMA{k-1}[O(\log n)]}$ which is also 
referred to as  $\THETA k$. 
A large collection of $\THETA 2$-complete problems is given in 
\cite{jcss/Krentel88,mst/GasarchKR95,DBLP:conf/fct/EiterG97}. 
Sections \ref{sect:delta2} and \ref{sect:theta2} in this work will also be devoted to 
$\THETA 2$-completeness results. 
Several problems complete for $\THETA k$ with $k \geq 3$ can be found in \cite{tcs/Krentel92}.
Hardness results in all these classes are obtained via log-space reductions, 
and we write $A\le B$ when a problem $A$ is log-space reducible to a problem $B$.

\section{$\DELTA 2$- and $\THETA 2$-completeness}
\label{sect:delta2}

In this work, we prove several new  $\DELTA 2$- and $\THETA 2$-completeness results. 
Problems in these classes naturally arise from optimization variants of $\NP$-complete problems,
where a sequence of oracle calls for the underlying $\NP$-complete problem is required to compute 
the optimum. $\THETA 2$-completeness applies, if binary search with logarithmically many oracle calls  suffices to find the optimum. Recently, an additional intuition of $\THETA 2$-complete
(and, more generally, $\THETA k$-complete) problems has been presented in \cite{DBLP:journals/tcs/LukasiewiczM17}, namely counting the number of positive instances in a set of instances of some $\NP$-complete (or, more generally, $\SIGMA{k-1}$-complete) problem. However, 
the problems for which we establish $\DELTA 2$- and $\THETA 2$-completeness here fall into the 
category of optimization problems derived from $\NP$-complete problems.

The base problems for proving our new $\DELTA 2$- and $\THETA 2$-completeness results
are the following: 
The following problems are considered as prototypical for the classes 
$\DELTA 2$ and $\THETA 2$. In particular, $\THETA 2$-hardness of \loglexmax will be used as starting
point for our reductions towards $\THETA 2$-hardness of \cardminsat\ for Krom formulas in Section~\ref{sect:theta2}.

\dproblem{\lexmax }{Propositional formula $\phi$ and an order $x_1 > \dots > x_\ell$  on the variables in $\phi$.}{Is $x_\ell$ true in the lexicographically maximal model of $\phi$?}

\dproblem{\loglexmax }{Propositional formula $\phi$ and an order $x_1 > \dots > x_\ell$  on some of the variables in $\phi$ with $\ell \leq \log |\phi|$.}{Is $x_\ell$ true in the lexicographically maximal bit 
vector $(b_1, \dots, b_\ell)$ that can be extended to a model of $\phi$?}

The \lexmax  problem will serve as our prototypical $\DELTA 2$-complete problem while the 
\loglexmax  problem will be the basis of our $\THETA 2$-completeness proofs. 
The $\DELTA 2$-completeness of \lexmax is stated in \cite{jcss/Krentel88} -- without proof though (see Theorem~3.4 in \cite{jcss/Krentel88}). The $\DELTA 2$-membership is easy. For the
$\DELTA 2$-hardness, the proof is implicit in a sequence of lemmas and theorems in  
\cite{jcss/Krentel88}. However, the main goal in \cite{jcss/Krentel88} is to advocate 
a new machine model (so-called $\NP$ {\em metric Turing machines}) for defining new 
complexity classes of optimization problems (the so-called $\OPT$ and $\OPT[z(n)]$ classes).
The \loglexmax problem is not mentioned explicitly in \cite{jcss/Krentel88} 
but, of course, it is 
analogous to the \lexmax  problem. 

To free the reader from the burden of tracing the line of argumentation in \cite{jcss/Krentel88} via several lemmas and theorems on the $\OPT$ and $\OPTlog$ classes, 
we give direct proofs of the $\DELTA 2$- and $\THETA 2$-completeness 
of \lexmax and \loglexmax, respectively, 
in the standard terminology of oracle Turing machines (cf.\ \cite{Papadimitriou94}).
In the first place, we thus have to 
establish the connection between oracle calls and optimization.
To this end, we  introduce the following problems:

\dproblem{\npmax }{Turing machine $M$ running in non-deterministic polynomial time and producing a binary string as output; string $x$ as input to $M$.}
{Let $w$ denote the lexicographically maximal output string 
over all computation paths of $M$ on input $x$; does the last bit of $w$  have value 1?}

\dproblem{\lognpmax }{Turing machine $M$ running in non-deterministic polynomial time and producing a binary string, whose length is logarithmically bounded in the size of the Turing machine and the input; string $x$ as input to $M$.}
{Let $w$ denote the lexicographically maximal output string 
over all computation paths of $M$ on input $x$; does the last bit of $w$  have value 1?}

\begin{thm}
\label{theorem:npmax}
The \npmax problem is $\DELTA 2$-complete and the \lognpmax problem is $\THETA 2$-complete.
\end{thm}
\begin{proof}
The $\DELTA 2$-{\em membership\/} of \npmax 
and the $\THETA 2$-{\em membership\/} of \lognpmax 
is seen by the following algorithm, 
which runs in deterministic polynomial time and has access to an $\NP$-oracle. 
The algorithm maintains a bit vector $(v_1, v_2,  \dots)$ of the lexicographically 
maximal prefix of possible outputs of TM $M$ on input $x$. To this end, we initialize $i$ to $0$ and ask
the following kind of questions to an $\NP$-oracle: 
Does there exist a computation path of TM $M$ on input $x$, 
such that the first $i$ output bits are $(v_1, \dots, v_i)$ and 
$M$ outputs yet another bit? If the answer to this oracle call 
is ``no'' then the algorithm stops with acceptance if ($i \geq 1$ and $v_i = 1$) holds
and it stops with rejection if ($i =0$ or $v_i = 0$) holds.

If the oracle call yields a ``yes'' answer, then our algorithm calls 
another $\NP$-oracle with the 
question: 
Does there exist a computation path of TM $M$ on input $x$, 
such that the first $i+1$ output bits are $(v_1, \dots, v_i, 1)$?
If the answer to this oracle call 
is ``yes'', then we set $v_{i+1} = 1$; otherwise we set $v_{i+1} = 0$. In either case, we then increment $i$ by 1 and continue with the first oracle question (i.e., 
does there exist a computation path of TM $M$ on input $x$, 
such that the first $i$ output bits are $(v_1, \dots, v_i)$ and 
$M$ outputs yet another bit?).

Suppose that the lexicographically maximal output produced by $M$ on input $x$ has $m$ bits. Then our algorithm needs in total $2m + 1$ oracle calls
 and the oracles work in non-deterministic polynomial time. 
Moreover, 
if the size of the output string of $M$ is logarithmically bounded, then the number of oracle calls is logarithmically bounded as well.  

For the {\em hardness\/} part, we first concentrate on the \npmax problem 
and discuss afterwards the modifications required to 
prove the corresponding complexity result also for the \lognpmax problem.
Consider an arbitrary problem ${\cal P}$ in $\DELTA 2$, i.e., ${\cal P}$
is decided in deterministic polynomial time by a Turing machine $N$ with access to an
oracle $\nsat$ for the SAT-problem.

Now let $x$ be an arbitrary instance of problem ${\cal P}$. From this we
construct an instance $M;x$ of \npmax, where we leave $x$ unchanged 
and we define $M$ as follows: 
In principle, $M$ simulates the execution of $N$ on input $x$. 
However, whenever $N$ reaches a
call to the SAT-oracle, with some input $\phi_i$ say, then $M$ 
non-deterministically executes $\nsat$ on $\phi_i$. In other words, in the computation tree of $M$, the subtree corresponding to this non-deterministic execution of $\nsat$ on $\phi_i$
is precisely the computation tree of $\nsat$ on $\phi_i$. On every computation path 
ending in acceptance
(resp.\ rejection) of $\nsat$, the TM $M$ writes 1 (resp. 0) to the output. After that, $M$ continues with the execution of $N$ as if it
had received a  ``yes'' (resp.\ a ``no'') answer
from the oracle call. 
After the last oracle call, $M$ executes $N$ to the end. If $N$ ends in acceptance, then $M$ outputs 1; otherwise it outputs 0 as the final bit.

It remains to prove the following claims: 
\begin{enumerate}
\item {\em Correctness.} Let $w$ denote the lexicographically maximal output string over all computation 
paths of $M$ on input $x$. Then the last bit of $w$  has value 1 if and only if 
$x$ is a positive instance of ${\cal P}$.
\item {\em Polynomial time.} The total length of each computation path of $M$ on input $x$ is polynomially bounded in 
$|x|$.
\item {\em Logarithmically bounded output.}
If the number of oracle calls of TM $N$ is logarithmically bounded 
in its input 
(i.e., problem ${\cal P}$ is in $\THETA 2$),
then 
the size of the output produced by $M$ on input $x$ is also 
logarithmically bounded.
\end{enumerate}

\noindent
{\em Correctness.} We first have to prove the following claim: for every $i \geq 1$
and for every bit vector $w_i$ of length $i$: 
$w_i$ is a prefix of the  
lexicographically maximal output string over all computation 
paths of $M$ on input $x$ if and only if $w_i$ encodes the 
correct answers of the first $i$ oracle calls of TM $N$ on input $x$,
i.e., for $j \in \{1, \dots, i\}$, 
$w_i[j] = 1$ (resp.\ $w_i[j] = 0$) if the $j$-th oracle call yields a ``yes''
(resp.\ a ``no'') answer.
This claim can be easily verified by induction on $i$. Consider the 
induction begin with $i = 1$. We have to show that the first bit of the output string is 1 
if and only if the first oracle call yields a ``yes'' answer. Indeed, suppose that the first oracle does yield a ``yes'' answer. 
This means that {\em at least one\/} of the computation paths of the non-deterministic oracle machine $N$ ends with acceptance. 
By our construction of $M$, there exists at least one computation path of the non-deterministic TM $M$ on which value 1 is written as first bit to the output. Hence, the string $w_1 = $ '1' is indeed the prefix of length 1 of the 
lexicographically maximal output string over all computation 
paths of $M$ on input $x$. 
Conversely, suppose that the first oracle call yields a ``no'' answer. This means that 
{\em all\/} of the computation paths of the non-deterministic oracle machine $N$ end with rejection. 
By our construction of $M$, all computation paths of the non-deterministic TM $M$ write 0 as the first bit to the output. Hence, the string $w_1 = $ '0' is indeed the prefix of length 1 of the 
lexicographically maximal output string. The induction step is shown by the same kind of reasoning.

Now let $m$ denote the number of oracle calls carried out by TM $N$ on input $x$
and let $w$ denote the lexicographically maximal output of 
TM $M$ over all its computation paths.
Then $w$ has length $m+1$ such that the first $m$ bits encode the correct
answers of the oracle calls of TM $N$ on input $x$. Moreover, 
by the construction of $M$, we indeed have that the last bit of $w$ is 1 
(resp.\ 0), if and only if $x$ is a positive (resp.\ negative) 
instance of ${\cal P}$.

\smallskip

\noindent
{\em Polynomial time.} Suppose that $N$ on input $x$ with $|x| = n$ 
is guaranteed to hold after at most $p(n)$ steps and 
that $\nsat$ on input $\phi$ with $|\phi| = m$ is 
guaranteed to hold after at most $q(m)$ steps for polynomials $p()$ and $q()$.
Then the total length of the computation of $N$ counting also the computation steps 
of the oracle machine $\nsat$ is bounded by a polynomial $r(n)$ 
with $r(n) = O(p(n) * q(p(n)))$. 
To carry this upper bound over to every branch of 
the computation tree of $M$ on input $x$ we have to solve a subtle problem: The upper bound $r(n)$ on the execution length of $M$ on input $x$ 
applies to every computation path where for every oracle call $\phi_i$, 
a correct computation path of $\nsat$ on input $\phi_i$ is simulated. However, in our simulation of $N$ with oracle $\nsat$ by a non-deterministic computation of $M$ on input $x$, 
we possibly produce computation paths which $N$ on input $x$ can never reach, e.g.: if the correct answer of $\nsat$ on oracle input $\phi_1$ is ``yes'', then the continuation of the 
simulation of $N$ on input $x$ on all computation paths where answer ``no''
on oracle input $\phi_1$ is assumed, can never be reached by the computation of $N$ (with oracle $\nsat$) on input $x$. To make sure that the 
polynomial upper bound applies to {\em every\/}
computation path of $M$, we have to extend TM $M$ by a counter such that 
$M$ outputs $0$ and halts if more than $r(n)$ steps have been executed. 

\smallskip
\noindent
{\em Logarithmically bounded output.}
Suppose that ${\cal P}$ is an arbitrary problem in $\THETA 2$ and that, therefore, 
TM $N$ only has logarithmically many oracle calls. 
By similar considerations as for the 
polynomial time bound of the computation of $M$ on $x$, 
we can make sure that the size of the output on every computation path of $M$ on $x$ is logarithmically bounded. To this end, 
we add a counter also for the number of oracle calls. 
The logarithmic bound (say $c \log n + d$ for constants $c,d$ and 
$n = |x|$) on the size of the output applies 
to every computation path, where for every oracle call $\phi_i$, 
a correct computation path of $\nsat$ on input $\phi_i$ is simulated by $M$. 
By adding to $M$ a counter for the size of the output, we can modify $M$ in such a way that $M$ outputs 0 and stops if the number $c \log n + d$ of 
output bits of $M$ (which corresponds to the number of oracle calls of $N$) is exceeded.
\end{proof}

The above problems will allow us now to prove the $\DELTA 2$- and $\THETA 2$-completeness 
of the problems \lexmax and \loglexmax, respectively.
As mentioned above, credit for these results (in particular, the $\DELTA 2$-completeness of \lexmax) goes to 
Mark W.~Krentel \cite{jcss/Krentel88}. 
However, we hope that sticking to the standard terminology of oracle Turing machines 
(and avoiding the ``detour'' via the $\OPT$ and $\OPTlog$ classes based on a new machine model) will help to better convey the intuition of these results.

\begin{thm}
\label{theorem:lexmaxsat}
The \lexmax problem is $\DELTA 2$-complete and the
\loglexmax problem is $\THETA 2$-complete.
The hardness results hold even if 
$\phi$ is in 3-CNF.
\end{thm}
\begin{proof}
The $\DELTA 2$-{\em membership\/} of \lexmax 
and the $\THETA 2$-{\em membership\/} of \loglexmax 
is seen by the following algorithm, 
which runs in deterministic polynomial time and has access to an $\NP$-oracle.
Let $m$ denote the number of variables in $\phi$ (in case of the \lexmax
problem) or the number of variables for which an order is given
(in case of the \loglexmax problem).
The algorithm maintains a bit vector $(v_1, \dots,v_m)$ of the lexicographically 
maximal (prefix of a possible) model of $\phi$. To this end, we initialize $i$ to $0$ and ask
the following kind of questions to an $\NP$-oracle: 
Does there exist a model of $\phi$ such that 
the truth values of the first $i$ variables $(x_1, \dots, x_i)$ 
are $(v_1, \dots, v_i)$ and variable $x_{i+1}$ is set to 1?
If the answer to this oracle call 
is ``yes'', then we set $v_{i+1} = 1$; otherwise we set $v_{i+1} = 0$. In either case, we then increment $i$ by 1 and continue with the next oracle call.

When $i = m+1$ is reached, then the algorithm checks the value of $v_m$ and 
stops with 
acceptance if $v_m = 1$ and it stops with rejection if $v_m = 0$
holds. 
If the number $m$ of variables of interest (i.e., those
for which an order is given) is logarithmically bounded in the size of $\phi$, then also the number of oracle calls of our algorithm is logarithmically bounded in the size of $\phi$.

\newcommand{\blank}{\sqcup}
\newcommand{\ssym}{\triangleright}

The {\em hardness\/} proof is by reduction from \npmax 
(resp.\ \lognpmax) to the 
\lexmax (resp.\ \loglexmax)
problem. 
Let $M;x$ be an arbitrary instance of 
\npmax  (resp.\ \lognpmax). 
We make the following assumptions on the Turing machine $M$: 
let $M$ have two
tapes, where tape 1 serves as input and worktape while tape 2 is the dedicated output tape.
This means that 
tape 2 initially has the starting symbol $\ssym$ in cell 0 and blanks
$\blank$ in all further tape cells. At time instant 0, the cursor on tape 2 makes a 
move to the right and leaves $\ssym$ unchanged. 
At all later time instants, either the current symbol (i.e., blank $\blank$) and the cursor on tape 2 are both left unchanged or the current symbol is
overwritten by 0 or 1 and the cursor is moved to the right. 
Recall that as output alphabet we have $\{0,1\}$.

Following the 
standard proof of the Cook-Levin Theorem
(see e.g.\ \cite{Papadimitriou94}), one can 
construct in polynomial time a propositional formula $\phi$
such that there is a one-to-one correspondence between the 
computation paths of the non-deterministic TM $M$ on input $x$ 
and 
the satisfying truth assignments of $\phi$.
Let $N$ denote the maximum length of a computation path of $M$ on 
any input of length $|x|$. Then 
formula $\phi$ is built from the following collection of variables: 
\newcommand{\sym}[4]{\textit{tape}_{#1}[#2,#3,#4]}
\newcommand{\cursor}[3]{\textit{cursor}_{#1}[#2,#3]}
\newcommand{\state}[2]{\textit{state}[#1,#2]}
\begin{description}
\item[$\sym{i}{\tau}{\pi}{\sigma}$] 
for $1 \leq i \leq 2$, 
$0\leq \tau\leq N$, 
$0\leq \pi\leq N$, and 
$\sigma\in \Sigma$, 
to express that at time instant $\tau$ of the computation, 
cell number $\pi$ of tape $i$ contains symbol $\sigma$, 
where $\Sigma$ denotes the set of all tape symbols of TM $M$.

\item[$\cursor{i}{\tau}{\pi}$] 
for $1 \leq i \leq 2$, 
$0\leq \tau \leq N$, and 
$0\leq \pi \leq N$, 
to express that at time instant $\tau$, 
the cursor on tape $i$ 
points to cell number $\pi$.

\item[$\state{\tau}{s}$] 
for $0\leq \tau\leq N$ and 
$s \in K$,
to express that at time instant $\tau$, 
the NTM $T$ is in state $s$, 
where $K$ denotes the set of all states of TM $M$.
\end{description}

\noindent
Suppose that $m$ denotes 
the maximum length of output strings of $M$ on any input of length $|x|$. 
Then we 
introduce additional propositional variables 
$\vec{x} = (x_1, x'_1, x_2, x'_2,\dots, x_m, x'_m)$ and $z$ and
we construct the propositional formula 
$\psi = \phi \wedge \chi$, where $\chi$ is defined as follows:

\begin{eqnarray*}
\chi & =  & \bigwedge_{\pi=1}^m x_{\pi} \leftrightarrow \sym{2}{N}{\pi}{1} 
\ \wedge \mbox{} \\
 & \wedge & \bigwedge_{\pi=1}^m x'_{\pi} \leftrightarrow \neg \sym{2}{N}{\pi}{\blank} 
\ \wedge \mbox{} \\
 & \wedge  & z \leftrightarrow \bigvee_{\pi=1}^N  \Big( x_\pi \wedge
       \bigwedge_{\pi'=\pi+1}^{N}  \neg x'_{\pi} \Big)
\end{eqnarray*}
In other words, $\chi$ makes sure that the variables in $\vec{x}$ encode the output string
and $z$ encodes the last bit of the output string, i.e.: for every $i$, variable $x_i$ is true in a
model $J$ of $\psi$ if the $i$-th bit of the output string (along the computation path of 
$M$ corresponding to $J$ when considered as a model of $\phi$) is 1. 
Consequently, truth value false 
of $x_i$ can mean that, on termination of $M$, 
the symbol in the $i$-th cell of tape 2 is  either 0 or $\blank$. 
The latter two cases are distinguished by the truth value of variable $x'_i$, which is 
false if and only if 
the symbol in the $i$-th cell of tape 2 is  $\blank$. 
Variable $z$ is true in a
model $J$ of $\psi$ if the last bit of the output string is 1. The last bit 
(in the third line of the above definition of formula $\chi$, this is position $\pi$) is recognized by the fact that beyond it, all cells on tape 2 contain the blank symbol.
Note that the truth value of $z$ in any model of $\psi$ is uniquely determined by the 
truth values of the $\vec{x}$ variables.

Finally, we transform $\psi$ into $\psi^*$ in 3-CNF by some standard transformation 
(e.g., by the Tseytin transformation \cite{tseytin}). We thus introduce additional propositional variables
such that every model of $\psi$ can be extended to a model of $\psi^*$ and every model of $\psi^*$ is also a model of $\psi$.
Now let $\vec{y}$ denote the vector
of variables 
$\sym{i}{\tau}{\pi}{\sigma}$, $\cursor{i}{\tau}{\pi}$, 
and $\state{\tau}{s}$ plus the additional variables introduced by the transformation into 3-CNF. 
Let the variables in $\vec{y}$ be arranged in arbitrary order and
let  $\ell$ denote the number of variables in $\vec{y}$. 

In the reduction from \npmax to \lexmax, we 
define the following order on all variables in $\psi^*$: 
$x_1 > x'_1 > \dots > x_m > x'_m > y_1 > \dots > y_\ell > z$. 
In the 
reduction from \lognpmax to \loglexmax, we 
define an order only on part of the variables in $\psi^*$, namely: 
$x_1 > x'_1 > \dots > x_m > x'_m > z $ (i.e., we ignore the variables in $\vec{y}$\,).

We now construct a particular model $J$ of $\psi^*$ 
for which we will then show that it is in fact the lexicographically maximal model of $\psi^*$.
Let $w$ denote the lexicographically maximal output string produced by $M$ on input $x$. 
Every computation path of $M$ corresponds to one or more models of $\psi^*$.
Consider the truth assignment $J$ on the variables 
$\vec{x}$, $\vec{y}$, and $z$ obtained as follows: 
$J$ restricted to $\vec{y}$ is chosen such that it is 
lexicographically  maximal among all truth assignments on $\vec{y}$ 
corresponding to a computation path of $M$ on input $x$ and with output $w$.
Note that the truth assignments of $J$ on $\vec{x}$ and $z$ 
are uniquely defined by the output $w$ of $M$ 
(i.e., no matter which concrete truth assignment on $\vec{y}$ to encode a computation path of $M$ with this output we choose).
 
We claim that $J$ is the lexicographically maximal model of $\psi^*$. 
To prove this claim, suppose to the contrary that there exists a lexicographically bigger model 
$J'$ of $\psi$. Then we distinguish 3 cases
according to the group of variables where $J'$ is bigger than $J$: 

(1) If $J'$ is bigger than $J$ on $\vec{x}$, then (since the truth values of $\vec{x}$ encode the output string of some computation of $M$ on $x$) there exists a bigger output than $w$. This contradicts our assumption that 
$w$ is maximal. 
(2) If $J'$ coincides with $J$ on $\vec{x}$ and $J'$ is bigger than $J$ on $\vec{y}$, then $J'$ restricted to 
$\vec{y}$ corresponds to a computation path producing the same output string $w$ 
as the computation path encoded by $J$ on $\vec{x}$. This contradicts our choice of truth assignment $J$ on $\vec{y}$.
(3) The truth value of $z$ in any model of $\psi^*$ is 
uniquely determined by the truth value of $\vec{x}$. 
Hence, it cannot happen that $J'$ and $J$ coincide on $\vec{x}$ but differ on $z$.

From the correspondence between the lexicographically maximal output string $w$ of $M$ on input $x$ and the 
lexicographically maximal model $J$ of $\psi^*$, the correctness of our problem reductions 
(both, from \npmax to \lexmax and from \lognpmax to \loglexmax) follows immediately, namely: 
For the \lexmax problem: the last bit 
in the lexicographically maximal output string $w$ of $M$ on input $x$ is 1 if and only if 
the truth value of variable $z$ in the lexicographically maximal model $J$ of $\psi^*$ is 1 (i.e., true).
Likewise, for the \loglexmax problem: 
the last bit 
in the lexicographically maximal output string $w$ of $M$ on input $x$ is 1 if and only if 
the truth value of variable $z$ in the lexicographically maximal truth assignment on 
$(x_1, x'_1,  \dots, x_m, x'_m, z)$
that can be extended to a model $J$ of $\psi^*$ is 1 (i.e., true). Note that in 
case of the \loglexmax problem, we may indeed simply ignore the $\vec{y}$ variables 
because the
truth values of $\vec{x}$ and $z$ are uniquely determined by the output $w$ of $M$ -- independently of the concrete choice of truth assignments to the $\vec{y}$ variables.
\end{proof}

\section{$\THETA 2$-complete variants of SAT}
\label{sect:theta2}

In this section, we study two natural variants of \sat:

\dproblem{\cardmin}{Propositional formula $\phi$ and an atom $x_i$.}{ Is $x_i$ true in a cardinality-minimal model of $\phi$?}

\dproblem{\cardmax}{ Propositional formula $\phi$ and an atom $x_i$.}{Is $x_i$ true in a cardinality-maximal model of $\phi$?}

\noindent
Both problems will be shown $\THETA 2$-complete. 
We start with hardness for \cardmax.

\begin{lem}
\label{lemma:cardmaxsat}
\cardmax is $\THETA 2$-hard,
even for formulas in 3-CNF.
\end{lem}
\begin{proof}
The $\THETA 2$-hardness of \cardmax is shown by reduction from \loglexmax. 
Consider an arbitrary instance of \loglexmax, which 
is given by a propositional formula $\phi$ 
and an order $x_1 > \dots > x_\ell$ 
over logarithmically many variables from $\phi$. From this we construct
an instance of \cardmax, which will be defined by a
propositional formula $\psi$ and a dedicated variable in $\psi$, namely $x_\ell$.  
We
simulate  the lexicographical order
over the variables $x_1 > \dots > x_\ell$ 
by adding ``copies'' of each variable $x_i$, i.e., for every $i \in \{1, \dots, \ell\}$,  we introduce $2^{\ell-i}-1$ new variables 
$x_i^{(1)}, x_i^{(2)},  \dots, x_i^{(r_i)}$ with $r_i = 2^{\ell-i}-1$. 
The formula $\psi$ is now obtained from $\phi$
by adding the following subformulas.
We
add to $\phi$ the conjuncts 
$(x_i \leftrightarrow x_i^{(1)}) \wedge \dots
\wedge (x_i \leftrightarrow x_i^{(r_i)})$. 
Hence,  setting $x_i$ to true in a model
of the resulting formula forces us to 
set all its $2^{\ell-i}-1$ ``copies'' to true. 
Finally, for the remaining variables
$x_{\ell+1}, \dots, x_{n}$ in $\phi$, we add ``copies''
$x'_{\ell+1}, \dots, x'_{n}$ and further extend $\phi$ by
the conjuncts 
$(x_{\ell+1} \leftrightarrow \neg x'_{\ell+1}) \wedge
\dots \wedge 
(x_{n} \leftrightarrow \neg x'_{n})$ to make the cardinality of models
indistinguishable on these variables. 

Since $\ell \leq \log |\phi|$, this transformation 
of $\phi$ into $\psi$ is feasible in polynomial time. 
Also note that $\psi$ is in 3-CNF, whenever $\phi$ is in 3-CNF.
We claim that our problem reduction is correct, i.e., 
let $\vec{b} = (b_1, \dots, b_\ell)$
denote the lexicographically 
maximal vector that can be extended to a model of $\phi$.
We claim that $x_\ell$ is true in a model $I$ of $\phi$, s.t.\
$I$ is an extension of $\vec{b}$  
iff $x_\ell$ is true in a cardinality-maximal model of~$\psi$. 
In fact, even a slightly stronger result can be shown, namely: 
if a model of $\phi$ is an extension of $\vec{b}$, then $I$ can be 
further extended to a cardinality maximal model $J$ of $\psi$. 
Conversely, if $J$ is a cardinality maximal model of $\psi$, then 
$J$ is also a model of $\phi$ and $J$ extends $\vec{b}$.
\nop{*** kuerzen
Call the resulting formula $\psi$.
In other words, of the $2m-2n$ variables 
$x_{m+1}, \dots, x_{n},x'_{m+1}, \dots, x'_{n}$, exactly $m-n$ have to be contained in every model of $\psi$; i.e., the difference between the cardinalities of two models is only due to the truth values of the variables $x_{1}, \dots, x_{m}$. Again, we clearly have that $x_m$ is 
true in the lexicographically maximal vector $(x_1, \dots, x_m)$ that can be extended to a model of $\phi$ iff $x_m$ is true in a cardinality-maximal model of $\psi$.
****************}
\end{proof}

\medskip

Our ultimate goal in this section is to show that 
the $\THETA 2$-completeness of \cardmax 
and also of \cardmin hold even if 
restricted to the Krom case. 
In a first step, we reduce the \cardmax problem to a variant of 
the {\sc Independent Set} problem, which we define next.
Note that the standard reduction from 3SAT to {\sc Independent Set} \cite{Papadimitriou94} is not sufficient: suppose we are starting off with 
a propositional formula $\phi$ with $K$ clauses. Then, if $\phi$ is satisfiable, 
every maximum independent set selects precisely $K$ vertices. Hence, additional work is needed to preserve the information on the cardinality of the models of $\phi$.  
The variant of the 
{\sc Independent Set} problem we require here is as follows:
For the sake of readability, 
we first explicitly introduce a lower bound on the independent sets of interest and then we show how to encode this lower bound implicitly.

\dproblem{\indsetWithBound}{Undirected graph $G = (V,E)$, vertex $v \in V$, and positive integer $K$.}{Is $v$ in a maximum independent set $I$ of $G$ with $|I| \geq K$?}
Note that the question is not whether $v$ belongs to {\em some\/} independent set of size at least $K$, but whether it belongs to a {\em cardinality-maximum\/} independent set and whether this 
maximum value is $\geq K$. 
Without the restriction to a {\em cardinality-maximum\/} independent set, the problem would 
obviously be in $\NP$. With this restriction, the complexity increases, as we show next.


\begin{lem}
\label{lemma:indsetWithBound}
The \indsetWithBound prob\-lem is $\THETA 2$-hard.
\end{lem}
\begin{proof} We extend the standard reduction from 
3SAT to {\sc Independent Set} \cite{Papadimitriou94}
to a reduction from \cardmax to \indsetWithBound. 
The result then follows from Lemma~\ref{lemma:cardmaxsat}.

Let $(\phi, x)$ denote an arbitrary instance of \cardmax
with $\phi = c_1 \wedge \dots \wedge c_m$, s.t.\ each clause 
$c_i$ is of the form $c_i = l_{i1} \vee l_{i2} \vee l_{i3}$
where each $l_{ij}$ is a literal. Let $X = \{x_1, \dots, x_n\}$ denote
the set of variables in $\phi$. 
We construct an instance $(G,v,K)$ of \indsetWithBound 
where 
$K := L * m$ 
with $L$ ``sufficiently large'', e.g., $L:= n +1$;
$G$ 
consists of  $3 * K + n$ 
vertices
$V = \{l_{i1}^{(j)}, l_{i2}^{(j)}, l_{i3}^{(j)} \} \mid
1 \leq i \leq m$ and $1 \leq j \leq L\} \cup 
\{u_1, \dots, u_n\}$;
and $v:= u_i$, for $x = x_i$.
It remains to specify the  edges $E$ of $G$:

(1) For every $i \in \{1, \dots, m\}$ and 
every $j\in \{1, \dots, L\}$, $E$ contains 
edges $[l_{i1}^{(j)}, l_{i2}^{(j)}]$,
$[l_{i1}^{(j)}, l_{i3}^{(j)}]$, and
$[l_{i2}^{(j)}, l_{i3}^{(j)}]$, i.e., 
$G$ contains $L$ triangles
$\{l_{i1}^{(j)}, l_{i2}^{(j)}, l_{i3}^{(j)} \}$
with $j\in \{1, \dots, L\}$.

(2) For every $\alpha, \beta,\gamma, \delta$, s.t.\
$l_{\alpha \beta}$ and
$l_{\gamma\delta}$ are complementary literals, 
$E$ contains $L$ edges   
$[l_{\alpha \beta}^{(i)}, l_{\gamma\delta}^{(j)}]$
with $i,j\in \{1, \dots, L\}$.

(3) For every $i \in \{1, \dots, n\}$ and every
$\alpha, \beta$, if $l_{\alpha \beta}$ is 
of 
the form $\neg x_i$, then $E$ contains $L$ edges 
$[l_{\alpha \beta}^{(j)}, u_i]$ 
($j\in \{1, \dots, L\}$).
%

The intuition of this reduction is as follows: The $L * m$ triangles $\{l_{i1}^{(j)}, l_{i2}^{(j)}, l_{i3}^{(j)} \}$
correspond to $L$ copies of the 
standard reduction from 3SAT to {\sc Independent Set} \cite{Papadimitriou94}.
Likewise, the edges between complementary literals are part of this standard
reduction. It is easy to verify that for every independent set $I$
with $|I| \geq K$,
there also exists an independent set $I'$ with $|I'| \geq |I|$, s.t.\
$I'$ chooses from every copy of a triangle the ``same'' endpoint, i.e., 
for every $i \in \{1, \dots, m\}$, 
if $l_{i\alpha}^{(\beta)} \in I'$ and 
$l_{i\gamma}^{(\delta)} \in I'$, then 
$\alpha = \gamma$.

Since $L = n+1$, the desired lower bound $K = L * m$ on the 
independent set can only be achieved if exactly one vertex is chosen from each 
triangle. We thus get the usual correspondence between models of $\phi$ and
independent sets of $G$ of size $\geq K$. Note that this correspondence leaves
some choice for those variables $x$ 
where the independent set contains no vertex corresponding to the
literal $\neg x$.
In such a case, we assume that the variable $x$ is set to true since
in \cardmax, our goal is to maximize the number of variables set to true.
Then a vertex $u_i$ may be added to an independent set $I$ 
with $|I| \geq K$, only if no vertex 
$l_{\alpha\beta}^{(j)}$ corresponding to a 
literal $l_{\alpha\beta}$ of the form $\neg x_i$ has been chosen into the 
independent set. Hence, $u_i \in I$ if and only if $x_i$ is to true in the corresponding model of $\phi$. 
\end{proof}

\begin{thm}
\label{theorem:cardmaxTwoSAT}
\cardmax 
(resp.\ \cardmin)
is $\THETA 2$-complete even 
when formulas are restricted to 
Krom and moreover the clauses consist of negative 
(resp.\ positive)
literals only. 
\end{thm}
\begin{proof}
Membership 
proceeds by the classical binary search \cite{Papadimitriou94}
for finding the optimum, asking questions
like 
``Does there exist a model of size at least $k$?'' or 
``Does there exist a model of size at most $k$?''. 
With logarithmically many  such calls to  an $\NP$-oracle,
the maximal (resp.\ minimal) size $M$ of all 
models of $\phi$ can thus be computed, and we check in a final question to an $\NP$-oracle if 
$x_i$ is true in a model of size $M$. 

%

For hardness, 
we start with the case of 
\cardmax.
To this end, we first reduce 
\indsetWithBound 
to the following intermediate problem:
Given an undirected graph $G = (V,E)$ and vertex $v \in V$,
is $v$ in a maximum independent set $I$ of $G$?
The fact that this intermediate problem reduces to 
\cardmax 
then follows 
by 
expressing this problem via a Krom formula with 
propositional variables 
$V$
and clauses $\neg v_i \vee \neg v_j$ 
for every edge  $[v_i,v_j]$ in $E$.
Hence, we obtain hardness also for the case of Krom formulas with negative literals only.
%
%
%
%

Hence, let us turn to the first reduction and 
consider an arbitrary 
instance $(G,v,K)$ of \indsetWithBound
with $G = (V,E)$ and $v \in V$. 
We define the 
corresponding 
instance $(H,v)$ of the intermediate problem with $H = (V', E')$ 
where  $V' = V \cup U$ with $U = \{u_1, \dots, u_K\}$ for fresh vertices $u_i$,
and
$E' = E \cup \{[u_i,v_j] \mid 1 \leq i \leq K,
v_j \in V\}$. 
The additional edges in $E'$ make sure that an independent
set of $H$ contains either only vertices from $V$ or only vertices from $U$. 
Clearly, $H$ contains the independent set $U$ with $|U| = K$.
This shows
$\THETA 2$-hardness of \cardmax.

To show the
$\THETA 2$-hardness result for 
\cardmin restricted to Krom and positive literals, we now can give a reduction from \cardmax (restricted to
Krom and negative literals).
Given, such a 
\cardmax instance $(\varphi,x)$
we construct $(\hat{\varphi},x)$ where $\hat{\varphi}$ is given as 
\begin{align*}
& \bigwedge_{x \in \var(\varphi)}\big((x\vee x')\land (x\vee x'')\big) 
\land \\
& \bigwedge_{(\neg x\vee \neg y) \in \varphi}\!\!\!\!\!\!\!\! \big(
(x'\lor y')\land(x''\lor y')\land(x'\lor y'')\land(x''\lor y'')\big).
\end{align*}
The intuition of variables $x',x''$ is to represent that $x$ is assigned to false. We have the following observations:
(1) 
a cardinality-minimal model of $\hat{\varphi}$ either sets 
$x$ or jointly, $x'$ and $x''$, to true
(for each variable $x\in\var(\varphi)$);
i.e.\
it is of the form 
$\tau(I):= I \cup \{x',x'' \mid x\in \var(\varphi)\setminus I\}$
for some $I\subseteq \var(\varphi)$;
(2) for each $I\subseteq \var(\varphi)$, it holds that 
$I$ is a model of $\varphi$ iff 
$\tau(I)$ is a model of $\hat{\varphi}$;
(3) for each $I,J\subseteq \var(\varphi)$,
$|I|\leq |J|$ iff 
$|\tau(I)|\geq |\tau(J)|$.
It follows that $(\hat{\varphi},x)$ is a yes-instance of \cardmin iff
$(\varphi,x)$ is a yes-instance of \cardmax.
%
%
\nop{*** kuerzen
Hence, $v$ is contained in a maximum independent set $I$ of $G$ with $|I| \geq K$ iff $v$ is contained in a maximum independent set $I$ of $H$.
***************}
\end{proof}

\section{Applications: Belief Revision and Abduction}
\label{sect:belief-revision}

In this section, we make use of the hardness results for \cardminsat\ in the previous section, in 
order to show novel complexity results for problems from the field of knowledge representation 
when restricted to the Krom fragment and the combined Horn-Krom fragment (i.e.\ Krom formulas 
with at most one positive literal per clause).

\paragraph{Belief Revision.}
Belief revision aims at
incorporating a new belief, while changing as little as possible of the original beliefs. 
We assume that a belief set is given by a 
propositional formula $\psi$ and that the new belief is given by a propositional formula $\mu$.
Revising $\psi$ by $\mu$ amounts to restricting the set of models of $\mu$ to those models which are ``closest'' to the models of $\psi$. 
Several {\em revision operators\/} have been proposed  which differ in how they define what ``closest'' means. 
Here, we focus on the revision operators due to Dalal \cite{Dalal88} and Satoh \cite{fgcs/Satoh88}.

Dalal's operator 
measures the distance between the models of
$\psi$ and $\mu$ in terms of the {\em cardinality of
model change\/}, i.e.,  let 
$M$ and $M'$ be two interpretations and let $M \Delta M'$ denote the
symmetric difference between $M$ and $M'$, i.e., 
$M \Delta M' = (M \setminus M') \cup  (M' \setminus M)$.
Further, let
$|\Delta |^{min}(\psi, \mu)$ 
denote the minimum number of propositional variables on which the
models of $\psi$ and $\mu$ differ.
We define
$|\Delta |^{min}(\psi, \mu) :=
\min \{ | M \Delta M' | :   M
\in \mmod(\psi),  M' \in \mmod(\mu)\}$, 
where $\mmod(\cdot)$ denotes the models of a formula. 
Dalal's operator is now given
as: 
$\mmod (\psi \revdal \mu) = \{ M \in \mmod(\mu)   : \exists M' \in \mmod(\psi)
\; s. \, t. \,\, | M \Delta M'| = | \Delta |^{min}(\psi, \mu) \}$.

Satoh's operator 
interprets the 
minimal change in terms of {\em set inclusion\/}. 
Thus 
let $\Delta^{min}(\psi, \mu) =
min_{\subseteq} ( \{  M \Delta M'\, : \,  M \in \mmod(\psi), \, M' \in
\mmod(\mu)\})$, where the operator
$min_{\subseteq}(\cdot)$ applied to a set $S$ of sets selects those  
elements $s \in S$, s.t.\ $S$ contains no proper subset of $s$.
Then we define Satoh's operator as 
$\mmod (\psi \revsat \mu) = \{ M \in \mmod(\mu)  \, : \, \exists M' \in \mmod(\psi)
\; s. \, t. \,\,  M \Delta M' \in  \Delta^{min}(\psi, \mu) \}$. 

Let $\revR$ denote a revision operator with 
$r \in \{D, S\}$. We analyse two  well-studied 
problems in belief revision.

\dproblem{\implication }{Propositional formulas $\psi$ and $\mu$, and an atom $x$.}{ Does 
$\psi \revR \mu \models x$ hold?}

\dproblem{\mc}{ Propositional formulas $\psi$ and $\mu$, 
and model $M$ of~$\mu$.}{ Does 
$ M \models \psi \revR \mu $ hold?}

\smallskip

The complexity of these problems has been intensively studied \cite{ai/EiterG92,jcss/LiberatoreS01}. 
For arbitrary propositional formulas $\psi$
and $\mu$, both problems are $\THETA 2$-complete for Dalal's revision operator. 
For Satoh's revision operator, 
\mc (resp.\ \implication) is 
$\SIGMA 2$-complete (resp.\
$\PI 2$-complete).
Both \cite{ai/EiterG92} and
\cite{jcss/LiberatoreS01} have also investigated the complexity of some restricted cases (e.g., when the formulas are restricted to Horn). Below, we pinpoint the complexity of these problems in the Krom case. Our hardness results will subsume known hardness results for the Horn case as well.

\begin{thm}
\label{theorem:br-dalal-Horn}
\implication and \mc 
are
$\THETA 2$-complete for Dalal's revision operator even if the formulas $\psi$ and $\mu$ are restricted to both
Krom and Horn form.
\end{thm}
\begin{proof}
The membership even holds without any restriction \cite{ai/EiterG92,jcss/LiberatoreS01}.
Hardness is shown as follows: Given an instance $(\phi, x_0)$
of  \cardmin where $X= \{x_0, x_1,\dots, x_n\}$ is the set of variables
in $\phi$. 
By Theorem 
\ref{theorem:cardmaxTwoSAT}, this problem
is $\THETA 2$-complete even if 
$\phi$ is in Krom form with positive literals only. 
We define the following instances of 
\implication and \mc:

Let ${\cal C} = \{c_1, \dots, c_m\}$
denote the set of clauses in $\phi$.
Let
$X = \{x_0, \dots, x_n\}$
be a set of variables and 
let $y,z$ be
two further variables. 
We 
define
$\psi$ and $\mu$
as:

\begin{align*}
& \psi = \big(\bigwedge_{(x_i \vee x_j) \in {\cal C}} (\neg x_i \vee \neg x_j)\big)
\wedge y \\
&\mu = \big(\bigwedge_{i=1}^n x_i\big) 
\wedge (x_0 \vee \neg y) \wedge (\neg x_0 \vee y)
\end{align*}
Obviously, both $\psi$ and $\mu$ are from the Horn-Krom fragment and 
can be constructed efficiently from $(\varphi,x_0)$.

Note that $\mu$ has exactly  two models $M = M_1 = \{x_1, \dots, x_n\}$
and $M_2 = \{x_1, \dots, x_n, x_0, y\}$.
Let $k$ denote the size of a minimum model of $\phi$. We claim that 
$x_0$ is contained in a minimum model of $\phi$ iff
$M \models \psi \revdal \mu $
iff $\psi \revdal \mu \not\models y$. 
The second equivalence is obvious. 
We prove both directions of the first equivalence separately.

First, suppose that $\phi$ has a minimum model $N$ 
with   $x_0 \in N$. Then the minimum distance of $M_1$ from models of $\psi$ is $k$, which is 
witnessed by the model $I = X \setminus N \cup\{y\}$
of $\psi$ (note that $x_0\notin M_1$). It can be easily verified that $M_2$ does not have
smaller distance from any model of $\psi$.
At best, $M_2$ also has distance $k$, namely from 
any model $I_2$ of $\psi$ of the form $I_2 = (X \setminus N') \cup \{y\}$, s.t.\
$N'$ is a minimum model of $\phi$.

Now suppose that 
every model $N$ of $\phi$ with 
$x_0 \in N$ has size $\geq k+1$. 
Then the distance of $M_1$ 
from any model $I_1$ of $\psi$ is $k+1$, since $M_1 \Delta I_1$ contains
$y$ and at least $k$ elements from $\{x_1, \dots, x_n\}$. 
On the other hand, the distance of $M_2$ from models of $\psi$ is $k$ witnessed 
by any model $I_2 = (X \setminus N') \cup \{y\}$ of $\psi$ 
where $N'$ is a minimum model of $\phi$. 

In summary, we have thus shown that $M_2 \models \psi \revdal \mu$
is guaranteed to hold. But $M_1 \models \psi \revdal \mu$ holds if and only if 
there exists a minimum model $N$ of $\phi$ with $x_0 \in N$.
\nop{%
Then the minimum distance 
between models of $M_1$ from models of $\psi$ is also $k$: this distance is witnessed 
by the model $I = (X \setminus N) \cup \{y\}$ and model $M_2$ of $\mu$. 
If we consider model $M_1$ of $\mu$ instead of $M_2$ then variable $y$ is added to the symmetric difference.
This can be compensated by eliminating $x_0$ from the symmetric difference if and only if
$I$ is obtained from a minimum model $N'$ of $\phi$ with $x_0 \in N'$.
}
\end{proof}

Compared to the 
$\THETA 2$-hardness proof for Dalal's revision operator in the Horn 
case 
\cite{ai/EiterG92,jcss/LiberatoreS01}, our construction is 
much simpler, thanks to the previously established hardness result
for \cardmin. Moreover, it is not clear how the constructions 
from 
\cite{ai/EiterG92,jcss/LiberatoreS01}
can be adapted to the Horn-Krom case.

The above theorem states that, for Dalal's revision operator, 
the complexity of the 
\implication and \mc problem 
does not decrease even if the formulas are restricted to  
Krom and Horn form. In contrast, we shall show below that for 
Satoh's revision, the complexity of  \implication and \mc 
drops one level in the polynomial hierarchy if 
$\psi$ and $\mu$ are Krom. Hence, below, both the
membership in case of Krom form and the hardness 
(which even holds for the restriction to Horn and Krom form) need to be proved.
Also here, our hardness reductions differ substantially from those in
\cite{ai/EiterG92,jcss/LiberatoreS01}.

\begin{thm}
\label{theorem:br-satoh-Horn}
\implication is $\coNP$-complete 
and \mc is 
$\NP$-complete for Satoh's 
operator if the formulas $\psi$ and $\mu$ are restricted to 
Krom. Hardness holds even if $\psi$ and $\mu$ are further restricted to both Krom and Horn.
\end{thm}
\begin{proof}
For the membership proofs, recall the 
$\coNP$-membership and $\NP$-membership proof
of \implication and \mc for the Horn case in 
\cite{ai/EiterG92}, Theorem 7.2; resp.\
\cite{jcss/LiberatoreS01}, Theorem 20.
The key idea there is that, 
given a model $I$ of  $\psi$ and a model $M$ of $\mu$
the subset-minimality of $I \Delta M$ can be tested in polynomial time by
reducing this problem to a SAT problem
involving the formulas $\psi$ and $\mu$. 
The same idea holds for  Krom form.

The crucial observation there is that,
given a model $I$ of  $\psi$ and a model $M$ of $\mu$, one can 
check efficiently whether there exists a model 
$J$ of  $\psi$ and a model $N$ of $\mu$ with 
$J \Delta N \subset I \Delta M$. Indeed, let 
$\Var(\psi) \cup \Var(\mu) = \{x_1, \dots, x_n\}$ and 
let $\{y_1, \dots, y_n\}$ be
fresh, pairwise distinct variables.
Then a model $J$ of $\psi$ and a model $N$ of $\mu$
with 
$J \Delta N \subset I \Delta M$
exist iff for some variable $x_j \in I \Delta M$, the following 
propositional formula is satisfiable: 
%
%
\begin{equation}\label{eq:f}
\psi[x/y] \wedge \mu \wedge 
(y_j \leftrightarrow x_j) 
\wedge
\bigwedge_{x_i \not\in I\Delta M} 
(y_i \leftrightarrow x_i) 
\end{equation}
%
%
Here, $\psi[x/y]$ denotes the formula that we obtain from $\psi$ by
replacing every $x_i$ by $y_i$. If $\psi$ and $\mu$ are Horn
(and, likewise, if they are Krom), then 
(\ref{eq:f}) is Horn (resp.\ Krom) as well,
and hence
this satisfiability check is feasible in polynomial~time.

Hardness is shown by the following reduction 
from 3SAT respectively co-3SAT: 
Let $\phi$ be an arbitrary Boolean formula in 3-CNF over 
the
variables $X = \{x_1, \dots, x_n\}$, i.e., 
$\phi = c_1 \wedge \dots \wedge c_m$,
s.t.\ each clause $c_i$ is of the form 
$c_i = l_{i1} \vee l_{i2} \vee l_{i3}$, 
where the $l_{ij}$'s are literals over $X$.
Let $Y = \{y_1, \dots, y_n\}$, $A = \{a_1, \dots, a_m\}$,
$B = \{b_1, \dots, b_m\}$, and $\{d\}$ be sets of fresh, pairwise distinct propositional variables. 
We 
define
$\psi$ and $\mu$
as follows:
\begin{eqnarray*}
\psi &=&  \big(\bigwedge_{i=1}^n (\neg x_i \vee \neg y_i)\big)
\wedge \big(\bigwedge_{j=1}^m 
\bigwedge_{k=1}^3 ( l^*_{jk} \rightarrow \neg a_j)\big) \wedge 
\\
&&\big(\bigwedge_{j=1}^m (a_j \leftrightarrow b_j)\big)\\
\mu &=& \big(\bigwedge_{i=1}^n (x_i \wedge y_i) \big) 
\wedge \big(\bigwedge_{j=1}^m a_j\big) 
\wedge \big(\bigwedge_{j=1}^m ( b_j \rightarrow d)\big)
\end{eqnarray*}
%
%
%
%
where we set $l^*_{jk} = x_\alpha$ if $l_{jk} = x_\alpha$ 
for some $\alpha \in \{1, \dots, n\}$
and 
$l^*_{jk} = y_\alpha$ if $l_{jk} = \neg x_\alpha$.
Finally, we define $M = X \cup Y\cup A$.
Both $\psi$ and $\mu$ are from the Horn-Krom fragment and 
can be constructed efficiently from $\phi$.
We claim that $\phi$ is satisfiable iff
$M \models \psi \revsat \mu $ 
iff $\psi \revsat \mu \not\models d$. 

Observe that every model of $\mu$ must set the variables in $X\cup Y \cup A$ to true. 
The truth value of the 
variables in $B$ may be chosen arbitrarily. However, as soon as at least one $b_j \in B$ is 
set to true, we must set $d$ to true due to the last conjunct in 
$\mu$. Hence, $M$ is the only model of $\mu$ where $d$ is set to false. 
From this, the second equivalence above 
follows. Below we sketch the proof of  the first equivalence.

First, suppose that $\phi$ is satisfiable and let $G$ be a model of 
$\phi$. We set $I := 
\{x_i \mid x_i \in G\} \cup \{y_i \mid x_i \not \in G\}$.
Indeed, $I$ is a model of $\psi$. We claim that 
$I$ is a witness for 
$M \models \psi \revsat \mu $, i.e., 
$I \Delta M$ is minimal. 
To prove this claim, let $J$ be a model of $\psi$ and $N$ a model of $\mu$, 
s.t.\ $J \Delta N \subseteq I \Delta M$. It suffices to show that then 
$J \Delta N =  I \Delta M$ holds. To this end, we compare $I$ and $J$ first on $X \cup Y$ 
then on $A$ and finally on $B \cup \{d\}$: By the first group of conjuncts in $\psi$,
we may conclude from $J \Delta N \subseteq I \Delta M$ that $I$ and $J$ coincide on 
$X \cup Y$. In particular, both $I$ and $J$ are models of $\phi$. 
But then they also coincide on $A$ since the second group of conjuncts in $\psi$
enforces $I(a_j) = J(a_j) = $ false for every $j$. 
Note that $(I \Delta M) \cap (B \cup \{d\}) = \emptyset$. 
Hence, no matter how we 
choose $J$ and $N$ on $B \cup \{d\}$, we have 
$I \Delta M \subseteq J \Delta N$.

Now suppose that $\phi$ is unsatisfiable. 
Let $I$ be a model of $\psi$. 
To prove $M \not\models \psi \revsat \mu $, 
we show that $I \Delta M$ cannot be minimal. 
By the unsatisfiability of $\phi$, we know that $I$ 
(restricted to $X$) 
is not a model of $\phi$. Hence, there exists at least one clause $c_j$ which is
false in $I$.  
The proof goes by 
constructing $J, N$ with 
$J \Delta N \subset I \Delta M$. 
The crucial observation is that the symmetric difference $I \Delta M$ can be decreased
by setting $J(a_j) = J(b_j) = N(b_j) =$ true and $N(d) =$ true. 
\end{proof}


\paragraph{Abduction.}
Abduction is used to produce explanations for some observed 
manifestations. Therefore, one of its primary fields of application is 
diagnosis.
A {\em propositional abduction problem\/} (PAP)~$\pap$ consists of a
tuple $\tuple{V,H,M,T}$, where~$V$ is a finite set of
\emph{variables}, $H \subseteq V$ is the set of \emph{hypotheses}, $M
\subseteq V$ is the set of \emph{manifestations}, and~$T$ is a
consistent \emph{theory} in the form of a propositional formula. A set
$\s \subseteq H$ is a \emph{solution} (also called \emph{explanation})
to~$\pap$ if $T \cup \s$ is consistent and $T \cup \s \models M$
holds. 
A \emph{system diagnosis problem} can be represented by a 
PAP $\pap = \tuple{V,H,M,T}$ as follows. The theory~$T$
is the system description.  The hypotheses $H \subseteq V$ describe
the possibly faulty system components.  The manifestations $M
\subseteq V$ are the observed symptoms, describing the malfunction of
the system.  The solutions~$\s$ of~$\pap$ are the possible
explanations of the malfunction.

Often, one is not interested in
\emph{any} solution of a given PAP~$\pap$
but only in \emph{minimal} solutions, where
minimality is defined  w.r.t.\
some preorder~$\preceq$ on the powerset~$2^H$. Two natural preorders
are set-inclusion~$\subseteq$ and smaller cardinality denoted as
$\leq$. Note that allowing {\em any} solution corresponds to choosing ``$=$'' as the preorder.
In \cite{siamcomp/CreignouZ06} a trichotomy (of 
$\SIGMA 2$-completeness, $\NP$-completeness, and tractability) has been proved for the
{\sc Solvability} problem of propositional abduction, i.e., deciding if a PAP $\pap$ has at least one solution
(the preorder ``$=$'' has thus been considered). 
Nordh and Zanuttini  
\cite{ai/NordhZ08} 
have  identified 
many further restrictions which make the {\sc Solvability} problem tractable. 
All of the above mentioned preorders have been 
systematically investigated in 
\cite{jacm/EiterG95}.
A study of the  counting complexity of abduction with 
these 
preorders has been carried out by 
Hermann and Pichler 
\cite{jcss/HermannP10}.
Of course, if a PAP has {\em any} solution than it also has a $\preceq$-minimal solution. Hence, the 
preorder is only of interest for problems like the following one:

\dproblem{\relevance }{PAP $\pap = \tuple{V,H,M,T}$ and hypothesis $h \in H$.}{ Is $h$ relevant, i.e., does $\pap$
have a $\preceq$-minimal solution $\s$ with $h \in \s$?}

\smallskip
\noindent
Known results \cite{jacm/EiterG95} are as follows: 
The \relevance problem is $\SIGMA 2$-complete for preorders
$=$ and $\subseteq$ and $\THETA 3$-complete for preorder $\leq$.
Moreover, 
the complexity drops by one level in the 
polynomial hierarchy if the
theory is restricted to Horn. In \cite{siamcomp/CreignouZ06},
the Krom case was considered for the preorder $=$.  
For the preorder $\subseteq$, the Krom case was implicitly settled in \cite{jacm/EiterG95}. 
Indeed, an inspection of the $\NP$-hardness proof of the \relevance problem in the
Horn case reveals that $\NP$-hardness holds even if the theory is simultaneously
restricted to Horn {\em and\/} Krom (see the proof of Theorem 5.2 in \cite{jacm/EiterG95}). 
Below we show that also for the preorder 
$\leq$, the complexity in the Krom case matches the Horn case.

\begin{thm}
\label{theorem:abduction}
The 
$\leq$-{\sc Relevance}\xspace
problem 
for 
PAPs $\pap = \tuple{V,H,M,T}$ 
where the theory $T$ is Krom
is $\THETA 2$-complete. 
Hardness
holds even if the theory is restricted simultaneously to Horn and Krom.
\end{thm}
\begin{proof}
The membership proof is analogous to the corresponding
one in \cite{jacm/EiterG95} for the general case. 
The decrease of complexity compared with arbitrary theories
is due to the 
tractability of satisfiability testing in the Krom case. 
$\THETA 2$-hardness is shown by the following problem reduction 
from \cardmin.
Consider an
arbitrary instance $(\phi, x_i)$ of \cardmin.
By Theorem \ref{theorem:cardmaxTwoSAT},
we may assume that $\phi$ is in positive Krom form. 
Let $\phi = (p_1 \vee q_1) \wedge \dots \wedge (p_m \vee q_m)$
over variables $X = \{x_1, \dots, x_n\}$
and let $G = \{g_1, \dots,  g_m \}$ be a set of fresh, pairwise distinct variables.
We define the PAP $\pap = \tuple{V,H,M,T}$ as follows: 
  \begin{eqnarray*}
    V &=& X \cup G \\
    H & = & X \\
    M & = &  G\\
    T &=& \{ p_i \to g_i \mid 1 \leq i \leq m \} \cup \{ q_i \to g_i
    \mid 1 \leq i \leq m \}
  \end{eqnarray*}
It is easy to verify that 
the models of $\phi$ coincide with the solutions of $\pap$. Hence,
$x_i$ is in a minimum model of $\phi$ iff 
$x_i$ is in a minimum solution of $\pap$.
\end{proof}

We note that 
our 
$\THETA 2$-hardness reduction for $\leq$-\relevanceplain 
is much easier than the 
$\THETA 2$-hardness reduction in
\cite{jacm/EiterG95} for the Horn case. (In fact, the latter reduction
maps a certain MAXSAT problem to abduction, and it is not immediate how
this reduction can be adapted to work in the combined Horn-Krom case.)
Again the reason why our reduction is quite simple relies on the fact that we 
can  start from the more closely related problem
\cardmin for Krom.

\section{Complete Classification of \cardmin}
\label{sect:classification}

Since neither the full Krom fragment nor even the Krom \emph{and} (dual) Horn fragment
makes the problems investigated tractable
(especially the ones that are $\THETA{2}$-complete) it is worth making a step
further and studying   how much we have to restrict the syntactic form of the
Krom formulas to decrease the complexity. A key for such tractability results
is to go through a more fine-grained complexity study of \cardmin. To this aim
we
propose to investigate this problem within Schaefer's framework that we
introduce next.

\medskip

A \emph{logical relation} (or a \emph{Boolean relation}) of arity $k$ is a relation $R\subseteq\{0,1\}^k$. We
will refer below to the following binary relation,
$\relOr{2}=\{(0,1),
(1,0), (1,1)\}$.
By abuse of notation we do 
not make a difference between a relation and its predicate symbol. 
 A
\emph{constraint} (or \emph{constraint application}) $C$ is a formula
$R(x_1,\dots,x_k)$, where $R$ is a logical relation of arity $k$ and
$x_1,\dots,x_k$ are (not necessarily distinct) variables.
If $u$ and $v$ are two variables, then $C[u/v]$ denotes the constraint obtained from $C$ in replacing 
each occurrence of $v$ by $u$. If $V$ is a set of variables, then $C[u/V]$ denotes the result of substituting $u$ 
to every occurrence of every variable of $V$ in $C$.
 An assignment $I$ of
truth values to the variables \emph{satisfies} the constraint if
$\bigl(I(x_1),\dots,I(x_k)\bigr)\in R$.
%
A \emph{constraint language} $\Gamma$ is a finite set of logical
relations. A \emph{$\Gamma$-formula} $\varphi$ is a conjunction
of constraint applications using only logical relations from
$\Gamma$, and hence is a quantifier-free first-order formula. For single-element constraint languages
$\set R$, we often omit parenthesis and speak about $R$-formulas
instead of $\set R$-formulas. With
$\var(\varphi)$ we denote the set of variables appearing in
$\varphi$. A $\Gamma$-formula $\varphi$ is satisfied by a truth
assignment $I$ if $I$ satisfies all the constraints in $\varphi$, such an $I$ is then a \emph{model} of $\varphi$,
 its \emph{cardinality} refers to the number of variables assigned 1. 
We say that two
quantifier-free first-order formulas $\varphi$ and $\psi$ are
equivalent $(\varphi\equiv\psi$) if they have the same sets of
variables and of satisfying assignments. 
 Assuming a canonical order on the variables, we
can regard assignments as tuples in the obvious way, and say that
a quantifier-free first-order formula \emph{defines} or \emph{implements}  the 
logical
relation of its models. For instance the binary clause $(x_1\lor x_2)$
defines the relation $\relOr{2}$. This notion can be naturally extended to existentially-quantified 
formulas in considering their free variables. For instance the formula $\exists y (x_1\lor y)\land(x_2\lor \neg y)$
 defines (or implements) the  relation $\relOr{2}$ as well.
 
\medskip

Throughout the text we refer to different types of Boolean relations following
Schaefer's terminology~\cite{stoc/Schaefer78}.
We say that a Boolean relation~$R$ is \emph{Horn} (resp.  
\emph{dual Horn}) 
if
$R$ can be defined by a CNF formula that
is Horn 
(resp. dual Horn). 
A relation $R$ is \emph{Krom} if it can be defined by
a 2-CNF formula.
 A relation $R$ is \emph{affine}
 if it can be defined by an \emph{affine} formula, \ie, conjunctions of
XOR-clauses (consisting of an XOR of some variables plus maybe the constant 1)
--- 
such a formula may also be seen as a system of linear equations over the field GF$(2)$. 
A
relation is \emph{ \daffine} 
 if it is definable by 
a conjunction of clauses, each of them being either a
unary clause or a $2$-XOR-clause (consisting of an XOR of two variables plus maybe
the constant 1) --- such a conjunctive formula may also be seen as a set of
 a conjunction of equalities and disequalities between pairs of variables.
A
relation $R$ of arity $k$  is \emph{0-valid} (resp., \emph{1-valid}) if $0^k\in R$
(resp., $1^k\in R$).
 Finally, a constraint language $\Gamma$ is Horn 
(resp. dual Horn, Krom, affine, \daffine, 0-valid, 1-valid) 
if every relation in $\Gamma$ is Horn 
 (resp. dual Horn, Krom, affine, \daffine).
 We say that a constraint language is \emph{Schaefer} if
$\Gamma$ is either Horn, dual Horn, Krom, or affine.

\medskip

The complexity study of the satisfiability of $\Gamma$-formulas, $\sat(\Gamma)$, started in 1978 in the seminal work of Schaefer.
He proved a famous dichotomy theorem: $\sat(\Gamma)$ is in $\PP$ if $\Gamma$ is either Schaefer, 0-valid or 1-valid, and    $\NP$-complete otherwise.
We study here the   following problem.

\dproblem{$\cardminsat(\Gamma)$}{ $\Gamma$-formula $\phi$ and 
atom $x$.}{Is $x$ true in a cardinality-minimal model of $\phi$?}


Our main result in this section is the following complete
classification within the Krom fragment.

\begin{thm}\label{theorem:classification}
Let $\Gamma$ be a  Krom constraint language.
 If $\Gamma$ is \daffine or Horn, then  $\cardminsat(\Gamma)$   is decidable in
polynomial time, otherwise it is $\THETA{2}$-complete.
\end{thm}

The following proposition covers the polynomial cases of Theorem \ref{theorem:classification}.

\begin{prop}\label{prop:classification_easy}
 Let $\Gamma$ be a  set of logical relations.
 If $\Gamma$ is \daffine or Horn, then  $\cardminsat(\Gamma)$   is decidable in
polynomial time. 
\end{prop}

\begin{proof}
 Let $\phi$ be a
$\Gamma$-formula. 
Suppose first that $\Gamma$ is Horn. In this case $\phi$ can be written as a
Horn formula. The unique minimal model of $\varphi$ can be found in polynomial time by
unit propagation.

Suppose now that $\Gamma$ is \daffine.
Without loss of generality we can suppose that $\phi$ does not contain unitary
clauses. Then each clause of $\phi$ expresses either the equality or the
disequality between two variables. Using the transitivity of the equality
relation and the fact that in the Boolean case $a \neq b \neq c$ implies $a =
c$, we can identify equivalence classes of variables such that each two
classes are either independent or they must have contrary truth values. We call
a pair $(A,B)$ of classes with contrary truth values \emph{cluster}, $B$ may be
empty. 
It follows easily that any two clusters are independent and thus to obtain a
model of $\phi$, we choose for each cluster $(A,B)$ either $A=1, B=0$ or $A=0,
B=1$. 
We suppose in the following that $\phi$ is satisfiable (otherwise, we will
detect a contradiction while constructing the clusters).
The weight contribution of each cluster to a model is either $|A|$ or $|B|$. It
is then enough to consider the cluster $(A,B)$ that contains the atom $x_i$. If
$x_i\in A$ and  
$|A|\le |B|$, then $x_i$ belongs to a cardinality-minimal model of $\varphi$, else it
does not. 
\end{proof}

To complete the proof of  Theorem \ref{theorem:classification} it remains to prove 
$\THETA{2}$-hardness for the remaining cases. These hardness results rely on the application of tools
from universal algebra (see, e.g.  
\cite{gei68,jeavons98,scs08,Nordhz09}), which  we now introduce.

Let us first recall a well-known closure operator  on sets of logical relations.

\begin{defi}\label{def:closure}
Let $\Gamma$ be a set of logical relations.
Then,  $\coclone\Gamma$ is the set of relations that can be defined by a formula of the form $\exists X\varphi$, such that $\varphi$ is a $(\Gamma\cup\set{=})$-formula and  $X\subseteq \var(\varphi)$.
\end{defi}

This closure operator is relevant in order to obtain complexity results for the satisfiability problem. Indeed, assume that $\Gamma_1\subseteq\coclone{\Gamma_2}$, then a $\Gamma_1$-formula can be transformed into a satisfiability-equivalent $\Gamma_2$-formula, thus showing that $\sat(\Gamma_1)$ can be reduced in polynomial time to $\sat(\Gamma_2)$ (see \cite{jeavons98}). Hence the complexity of  $\sat(\Gamma)$ depends only on $\coclone\Gamma$. 

The set $\coclone\Gamma$  is a relational clone (or \emph{co-clone}). 
Accordingly, in order to obtain a full
complexity classification for the satisfiability problem we only
have to study the co-clones.

Interestingly, there exists a Galois correspondence between the lattice of Boolean relations (co-clones) and the lattice of Boolean functions (clones) (see \cite{gei68,bokakoro69}). This  correspondence is established through the operators $\pol .$ and $\inv .$ defined below.

\begin{defi}\rm
Let $f\colon\{0,1\}^m\rightarrow\{0,1\}$ and $R\subseteq\{0,1\}^n$. We say that $f$ is a \emph{polymorphism} of $R$, if for all $x_1,\dots,x_m\in R$, where $x_i=(x_{i}[1],x_{i}[2],\dots,x_{i}[n])$, we have
$
\bigl(f\bigl(x_{1}[1],\cdots,x_{m}[1]\bigr),
    f\bigl(x_{1}[2],\cdots,x_{m}[2]\bigr),
    \dots,
    f\bigl(x_{1}[n],\cdots,x_{m}[n]\bigr)
  \bigr)\in R.
$
\end{defi}

For a set of relations $\Gamma$ we write $\pol\Gamma$ to denote
the set of all polymorphisms of $\Gamma$, i.e., the set of all Boolean
functions that preserve every relation in  $\Gamma$. For every $\Gamma$,
$\pol\Gamma$ is a \emph{clone}, i.e., a set of Boolean functions that contains
all projections   and is closed under composition. 
%
%
As shown first in
\cite{gei68,bokakoro69} the operators $\pol{.}-\inv{.}$ constitute a Galois
correspondence between the lattice of sets of Boolean relations and the
lattice of sets of Boolean functions. In particular for every set  $\Gamma$ of Boolean
relations 
$\inv{\pol{\Gamma}}=\coclone\Gamma$,
and  there is a
one-to-one correspondence between clones and co-clones. Hence
we may compile a full list of co-clones from the list of clones obtained by Emil Post in \cite{pos41}. The list of all Boolean clones with finite bases can be found e.g.\ in \cite{bocrrevo03}. A compilation of all co-clones with finite bases is given in \cite{borescvo05}. In the following, when discussing about bases for clones or co-clones we implicitly refer to these two lists.
 Figure~\ref{figure:postlattice}  provides a representation of the inclusion structure of the clones, and hence also of the co-clones.
For two clones $\mcal C_1$ and $\mcal C_2$, it holds that $\inv{\mcal C_1}\subseteq\inv{\mcal C_2}$ if and only if $\mcal C_2\subseteq\mcal C_1$.

\begin{figure}
  \begin{center}\includegraphics[scale=0.75]{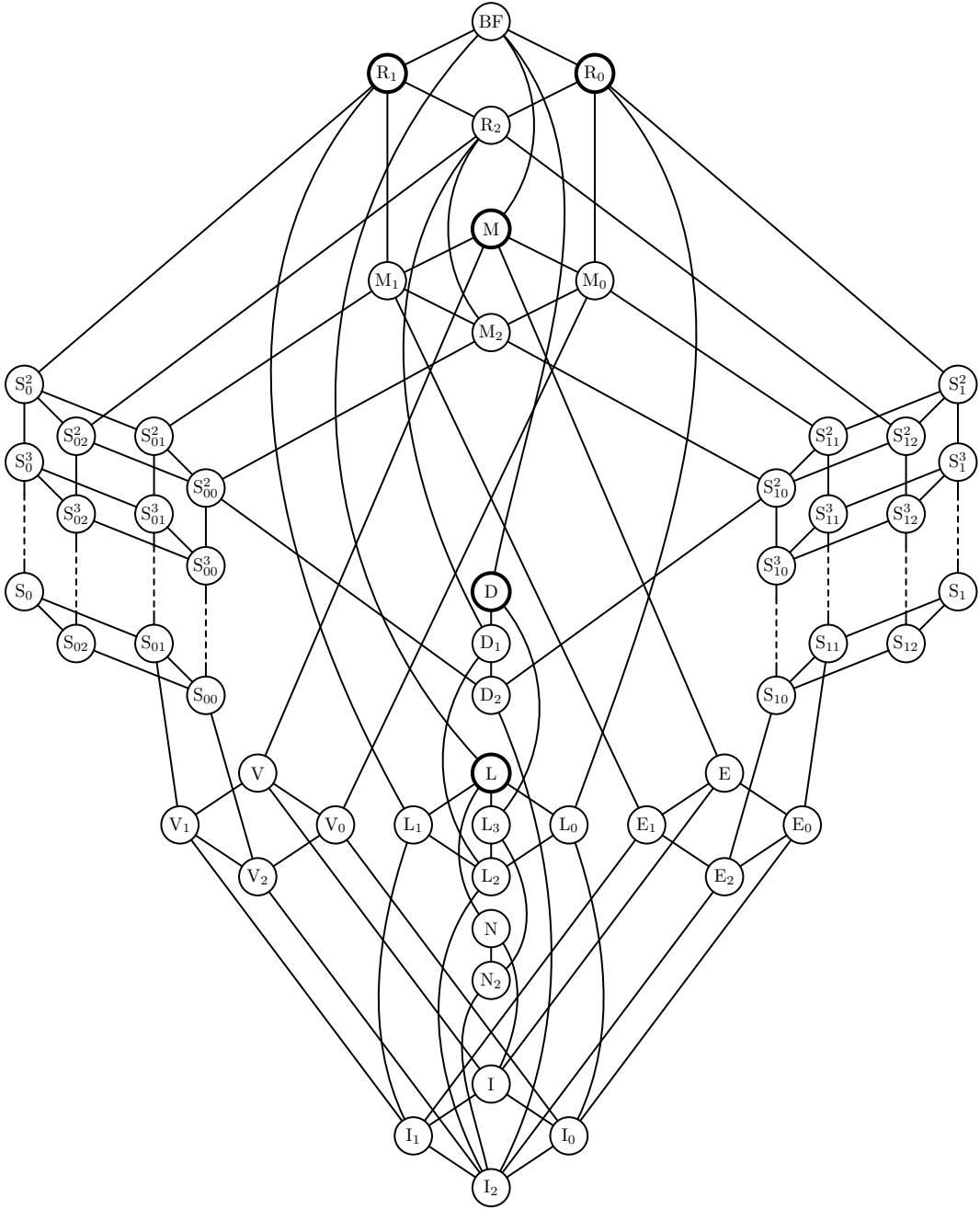}\end{center}
  \caption{Lattice of all Boolean clones}
  \label{figure:postlattice}
\end{figure}

\medskip

There exist easy criteria to determine if a given relation is Horn, dual Horn, Krom, affine or \daffine. 
Indeed all these  classes can be characterized by their polymorphisms (see \eg~\cite{CreignouV08} for a detailed description).
 We recall some of these properties here briefly for completeness. 
The operations of conjunction, disjunction,
addition and majority applied  on $k$-ary Boolean vectors are applied coordinate-wise.
\begin{itemize}
\item $R$ is Horn if and only if $m, m' \in R$ implies $m \land m'\in R$.
\item $R$ is dual Horn if and only if $m, m' \in R$ implies $m \lor m'\in R$.
\item $R$ is affine if and only if $m, m', m'' \in R$ implies $m\oplus m' \oplus m'' \in R$.
\item $R$ is Krom if and only if $m, m', m'' \in R$ implies $Majority(m,  m', m'') \in R$.
\item $R$ is \daffine   if and only if $m, m', m''\in R$ implies both  $m\oplus m' \oplus m'' \in R$ and  $majority(m,  m', m'') \in R$.
\end{itemize}

In terms of clones, given a set $\Gamma$ of logical relations, this corresponds to   the following:
\begin{itemize}
\item $\Gamma$  is Horn if and only if $\Gamma\subseteq \inv{\clonename{E}_2}$.
\item $\Gamma$  is dual Horn if and only if $\Gamma\subseteq \inv{\clonename{V}_2}$.
\item $\Gamma$  is affine if and only if $\Gamma\subseteq \inv{\clonename{L}_2}$.
\item $\Gamma$  is Krom if and only if $\Gamma\subseteq \inv{\clonename{D}_2}$.
\item $\Gamma$  is \daffine   if and only if $\Gamma\subseteq \inv{\clonename{D}_1}$.
\end{itemize}

\medskip

Unfortunately, since we are here interested in cardinality-minimal models, it seems difficult  to use the Galois connection explained above.
Indeed, existential variables and equality constraints that may occur when transforming  a $\Gamma_1$-formula into a satisfiability-equivalent $\Gamma_2$-formula are problematic, as they can change the set of models and the cardinality of each model. 
Therefore, we will use two other notions of closure, which are  more restricted.

First we need to introduce the notion of frozen variable 
 in a formula (resp. in a relation).

\begin{defi}\label{def:frozen}
Let $\varphi$ be a formula and
$x\in\var(\varphi)$, then $x$ is said to be \emph{frozen} in $\varphi$ if it
is assigned the same truth value in all its models.

Since a relation can be defined by a formula, by abuse of language we also speak about a frozen variable in a relation.
\end{defi}

Let us now define two closure operators on sets of logical relations.

\begin{defi}\label{def:closure:2}
Let $\Gamma$ be a set of logical relations.
\begin{itemize}
 \item $\cocloneFr{\Gamma}$ is the set of relations that can be defined by a formula  of the form   $\exists X\varphi$,  such that $\varphi$ is a $(\Gamma\cup\set{=})$-formula, $X\subseteq \var(\varphi)$ and every variable in $X$ is frozen in $\varphi$. 
 \item $\cocloneFrnoeq{\Gamma}$ is the set of relations that can be defined by a formula  of the form   $\exists X\varphi$,  such that $\varphi$ is a $\Gamma$-formula, $X\subseteq \var(\varphi)$ and every variable in $X$ is frozen in $\varphi$. 
\end{itemize}
\end{defi}

These two notions differ in the sense that the first one allows equality constraints, and the other does not. The symbol $\ne$ stands for ``no equality constraint is allowed''.

On the one hand,  the first notion of closure is well known from an algebraic point of view. Indeed, $\cocloneFr{\Gamma}$ is a partial frozen co-clone. The lattice of the partial frozen co-clones is   partially known, especially within  the Krom fragment (see\cite{Nordhz09}).
On the other hand, the more restricted notion of closure, $\cocloneFrnoeq{.}$,  is useful for complexity issues.
Indeed, this closure operator induces  
 reductions between problems we are interested in, as shown by the following proposition.

\begin{prop}\label{prop:closure_reduction}
  Let $\Gamma_1$ and $\Gamma_2$ be constraint languages with $\Gamma_1\subseteq\cocloneFrnoeq{\Gamma_2}$. Then
  $\cardminsat(\Gamma_1)\le\cardminsat(\Gamma_2)$.
\end{prop}

\begin{proof}
Let $\phi_1$ be a $\Gamma_1$-formula. We construct a formula $\phi_2$
by performing the following steps:
\begin{itemize}
\item Replace in $\phi_1$ every constraint from $\Gamma_1$ by its defining
existentially quantified $\Gamma_2$-formula, in which all existential variables are frozen. Use fresh existential variables for each constraint.
\item Delete existential quantifiers. 
\end{itemize}
Then, obviously, $\phi_2$ is a $\Gamma_2$-formula and
$\phi_1$ is satisfiable if and only if $\phi_2$ is satisfiable. 
Moreover, since all existentially quantified  variables are frozen,  removing the quantifiers   preserves the
cardinality-minimal models, i.e., there is a one-to-one correspondence between the cardinality-minimal 
models of $\phi_1$ and the ones of $\phi_2$, which preserves the truth values of the variables in $\var(\phi_1)$. 
Therefore, a given atom $x$ is true in a cardinality-minimal model of $\phi_1$ if and only if it is true in a cardinality-minimal model of $\phi_2$.
Moreover the
complexity of the above transformation is polynomial, thus concluding the proof.
\end{proof}

As a consequence to get hardness results  in the Krom fragment we will use both notions of closure.
Roughly speaking the main strategy is as follows:   Theorem \ref{theorem:cardmaxTwoSAT} shows that  $\cardminsat(\relOr{2})$ is $\THETA{2}$-hard. Hence, in most of the cases, in order to prove that 
for some class of constraint languages the problem \cardminsat\  is  $\THETA{2}$-hard, according to Proposition \ref{prop:closure_reduction} we will prove that  for every language $\Gamma$ in the studied class, $\cocloneFrnoeq{\Gamma}$ contains $\relOr{2}$. But showing that $\relOr{2}\in\cocloneFrnoeq{\Gamma}$ will be done using information on the corresponding partial co-clone $\cocloneFr{\Gamma}$. For this  we need an additional technical notion.  An $n$-ary relation $R$ is \emph{irredundant} if there is no pair  $(i,j)$,   $1\leq i<j\leq n$, such that for all $(\alpha_1,\ldots, \alpha_n)\in R$, $\alpha_i=\alpha_j$, and if there is no $i$, $1\leq i\leq n$, such that for all $(\alpha_1,\ldots,\alpha_i,\ldots, \alpha_n)\in R$, the tuple $(\alpha_1,\ldots,1-\alpha_i,\ldots, \alpha_n)\in R$ as well (which means that the $i$th variable in the formula representing $R$ is unconstrained)

The motivation for defining irredundant relations is that one can easily represent these relations with formulas in which no equality constraints appear. Hence,  as stated in the following lemma,  if we can define an irredundant relation using the $\cocloneFr{.}$-operator, we can also  obtain an implementation using the $\cocloneFrnoeq{.}$-operator.

\begin{lem}\label{lemma:frozen_implementation}
Let $\Gamma$ be a set of logical relations and $R$ be an irredundant relation that has no frozen variable. If $R\in\cocloneFr{\Gamma}$, then $R\in\cocloneFrnoeq{\Gamma}$.
\end{lem}

\begin{proof}
Let $R$ be such an irredundant relation, and suppose that it is defined by the formula $\exists X\varphi$ 
where $\varphi$ is a $(\Gamma\cup\{=\})$-formula in which all variables from $X$ are frozen.
Two variables from $\var(\varphi)\setminus X$ (this set corresponds to variables of $R$)   cannot  appear in an equality constraint in $\varphi$ since $R$ is irredundant. 
A variable from $X$ and a variable from $\var(\varphi)\setminus X$  cannot occur together in an equality constraint in $\varphi$ either since all variables from $X$ are frozen in $\varphi$ and   $R$ has no frozen variable. Therefore, equality constraints can only involve variables from $X$. Using the transitivity of the equality relation we can identify equivalence classes.
In the formula $\exists X\varphi$, all variables from $X$ that are  in the same equivalence class can be replaced by a single variable, 
 thus obtaining a frozen implementation of $R$ with no equality.
\end{proof}

We are now in a position to prove the following proposition, which covers the hard cases of Theorem \ref{theorem:classification}.

\begin{prop}\label{prop:classification_hard}
 Let $\Gamma$ be a  Krom constraint language.
 If $\Gamma$ is neither \daffine nor Horn, then  $\cardminsat(\Gamma)$   is $\THETA{2}$-hard. 
\end{prop}

\begin{proof}

Observe that while the Galois connection  explained above   does not seem to be
useful to study the complexity of  $\cardminsat(\Gamma)$, nevertheless   the obtained classification obeys the borders among co-clones and the classification will be obtained through
a case study based on Post's lattice.
  Since $\Gamma$ is Krom, $\Gamma\subseteq \inv{\clonename{D}_2}$. Since $\Gamma$ is neither \daffine nor Horn, $\Gamma\subseteq \inv{\clonename{D}_1}$ and $\Gamma\subseteq \inv{\clonename{E}_2}$.
Therefore, according to Post's lattice (see Figure \ref{figure:postlattice}) there are only five cases to consider, either $\coclone{\Gamma}=\inv{\clonename{D}_2}$, 
$\inv{\clonename{S}_{00}^2}$,
 $\inv{\clonename{S}_{01}^2}$,
 $\inv{\clonename{S}_{02}^2}$ 
or $\inv{\clonename{S}_0^2}$.

\paragraph{Case $\coclone{\Gamma}= \inv{\clonename{S}_{00}^2}$ or $\coclone{\Gamma}= \inv{\clonename{S}_{02}^2}$ or  $\coclone{\Gamma}= \inv{\clonename{S}_{0}^2}$.} 
According to \cite[Theorem 15]{Nordhz09}, the three co-clones $\inv{\clonename{S}_{00}^2}$,  $\inv{\clonename{S}_{02}^2}$ and $\inv{\clonename{S}_{0}^2}$ are covered by a single partial frozen co-clone. Therefore, if $\coclone{\Gamma}=\inv{\clonename{S}_{00}^2}$ (respectively, $\coclone{\Gamma}=\inv{\clonename{S}_{02}^2}$, $\coclone{\Gamma}=\inv{\clonename{S}_{0}^2}$), then 
$\cocloneFr{\Gamma}=\coclone{\Gamma}=\inv{\clonename{S}_{00}^2}$ (respectively, $\cocloneFr{\Gamma}=\coclone{\Gamma}=\inv{\clonename{S}_{02}^2}$, $\cocloneFr{\Gamma}=\coclone{\Gamma}=\inv{\clonename{S}_{0}^2}$).  Since these three co-clones contain the relation $\relOr{2}$ (see e.g. \cite{borescvo05}), we obtain that  $\cocloneFr{\Gamma}$ contains $\relOr{2}$. Since the relation $\relOr{2}$ is irredundant and has no frozen variable, it follows from Lemma \ref{lemma:frozen_implementation} that $\relOr{2}\in\cocloneFrnoeq{\Gamma}$. We conclude  using Theorem \ref{theorem:cardmaxTwoSAT} and Proposition \ref{prop:closure_reduction}.

\paragraph{Case $\coclone{\Gamma}= \inv{\clonename{D}_{2}}$.} The co-clone $\inv{\clonename{D}_{2}}$ is not covered by a single frozen co-clone, but the structure of partial frozen co-clones covering this co-clone is known \cite[Theorem 19]{Nordhz09}. In particular it holds  that $\frclosure{\Gamma}\supseteq\Gamma_4^p$, where $\Gamma_4^p$ is a  partial frozen co-clone that contains the relation
$R^p_4$  defined
by $R^p_4(x_1, x_2, x_3, x_4)=(x_1\lor x_2)\land (x_1\neq x_3)\land(x_2\neq x_4)$. Since the relation $R^p_4$  is irredundant and has no frozen variable, by Lemma \ref{lemma:frozen_implementation} it follows that $R^p_4\in\cocloneFrnoeq{\Gamma}$.
We now prove that  $\cardminsat(\relOr{2})\le\cardminsat(R^p_4)$, thus concluding the proof according to Theorem \ref{theorem:cardmaxTwoSAT} and Proposition \ref{prop:closure_reduction}.

 Let
$\varphi=\bigwedge_{(i,j)\in E} \relOr{2}(x_i,x_j)$ be an $\relOr{2}$-formula with
$\var(\varphi)=\{x_1,\ldots, x_n\}$ and $E\subseteq\{1,\ldots , n\}^2$.
 Let us consider the formula $\varphi'$ built over the set of variables 
 \begin{align*}
   \{x_1,\ldots, x_n, x'_1,\ldots, x'_n, x''_1,\ldots, x''_n\},
 \end{align*}
where the $x'_i$'s and the $x''_i$'s are fresh variables, and defined by
$ \varphi'=\bigwedge_{(i,j)\in E} R_4^p(x_i, x_j, x'_i, x'_j)\land R_4^p(x''_i,
x''_j, x'_i, x'_j).$
Observe that in every model of $\varphi'$, for all $i$ we have $x_i=x''_i$ and
$x_i\ne x'_i$. Therefore, it is easily seen that there is a one-to-one
correspondence between the set of models of $\varphi$ and the set of models of
$\varphi'$, preserving the truth values of $x_1,\ldots, x_n,$ and
transforming every model of weight $k$ of $\varphi$ into a model of weight
$n+k$ of $\varphi'$. Therefore, a given atom $x_i$ is true in a
cardinality-minimal model of $\varphi$ if and only if it is true in a
cardinality-minimal model of $\varphi'$, thus concluding the proof.

\paragraph{Case $\coclone{\Gamma}= \inv{\clonename{S}_{01}^2}$.} The co-clone $\inv{\clonename{S}_{01}^2}$ is not covered by a single co-clone either and the structure of the partial frozen co-clones covering this co-clone is not known. Therefore, we have to use another technique. We prove directly that $\cardminsat(\relOr{2})\le \cardminsat(\Gamma)$, thus concluding the proof according to Theorem \ref{theorem:cardmaxTwoSAT}. 

Observe that all relations in $\Gamma$ are 1-valid and dual-Horn, and there exist at least one relation $S$ in $\Gamma$ which is not 0-valid and one relation $R$ that is  not Horn. Since $S$ is 1-valid but not 0-valid,  
$\vec 1\in S$ and $\vec 0\notin S$. Therefore, the $\Gamma$-constraint $C_1(x)=S(x,\ldots , x)$ defines the constant 1.

Consider the constraint $C=R(x_1,\ldots, x_k)$. Since $R$ is non-Horn there exist  $m_1$ and $m_2$ in $R$ such that $m_1\land m_2\notin R$. Since $R$ is 1-valid  and dual Horn, we have $\vec 1\in R$ and  $m_1\lor m_2\in R$. For $i,j\in\{0,1\}$, set $\displaystyle V_{i,j}=\{x\mid x\in V, m_1(x)=i\land m_2(x)=j\}$. Observe that $V_{0,1}\ne\emptyset$ (respectively, $V_{1,0}\ne\emptyset$), otherwise $m_1\land m_2=m_2$ (resp., $m_1\land m_2=m_1$), contradicting the fact that  $m_1\land m_2\notin R$. 
Consider the $\{R\}$-constraint: $M(w,x,y,t)=C[w/V_{0,0},\, x/V_{0,1},\, y/V_{1,0},\, t/V_{1,1}].$  According to the above remark the two variables $x$ and  $y$ effectively occur in this constraint. Let us examine the  set of models of $M$:  it contains $0011$ (since $m_1\in R$), $0101$ (since $m_2\in R$), $0111$ (since $m_1\vee m_2\in R$) and $1111$ (since $R$ is 1-valid), but it does not contain $0001$ (since by assumption $m_1\land m_2\notin R$). We make a case distinction according to whether $1001\in M$ or not. 

Suppose $1001\not\in M$. To every  $\relOr{2}$-formula $\varphi=\bigwedge_{i=1}^m \relOr{2}(x_i, y_i)$ we associate the $\Gamma$-formula
 $\phi'=\bigwedge_{i=1}^m M(\alpha_i, x_i, y_i, t)\land C_1(t)$ where $\alpha_i$ for $ i=1,\ldots m$ and $t$ are fresh variables. Observe that cardinality-minimal models of $\phi$ and  cardinality-minimal models of $\phi'$ coincide with $t=1$ and  $\alpha_i=0$ for $ i=1,\ldots m$, thus concluding the proof.

Suppose now that $1001\in M$. To every  $\relOr{2}$-formula $\varphi=\bigwedge_{i=1}^m \relOr{2}(x_i, y_i)$ with $\vert \var(\varphi)\vert =n$, we associate the $\Gamma$-formula
 $\phi'=\bigwedge_{i=1}^m \bigwedge_{j=1}^{n+1}M(\alpha_i^j, x_i, y_i, t)\land C_1(t)$ where $\alpha_i^j$ for $i=1,\ldots m, j=1,\ldots, n+1$ and $t$ are fresh variables. Observe that minimal models of $\varphi$ can be extended to minimal models of $\phi'$ in setting $t$ to 1 and all $\alpha_i^j$ to 0. Conversely, note that $\phi'$ has a model of  cardinality $n+1$, namely  $m$ such that    $m(t)=1$, $m(x)=1$ for all $x\in\var(\phi)$ and 
$m(\alpha_i^j)=0$ for $i=1,\ldots m, j=1,\ldots, n+1$. Therefore, every cardinality-minimal model  of $\phi'$ has cardinality less than or equal to $n+1$. So, given a cardinality-minimal model $m'$ of $\phi'$, $m'\models (x_i\lor y_i)$ for all $i$, otherwise this would imply $m'(\alpha_i^j)=1$ for all $j$, contradicting the minimality of $m'$. 
Therefore, $m'$ restricted to $\var(\phi)$ is a model of $\phi$, and even a cardinality-minimal model of $\phi$. This concludes the proof.
\end{proof}

\begin{rem}
The complexity classification of $\cardminsat$ in the Krom fragment has been obtained by  means of partial frozen co-clones.
The proof could also have been obtained   using  an even more restricted notion of closure, namely $\coclonenoexistnoeq{.}$, which  allows neither existential variables nor equality constraints, together with the notion of weak bases introduced in \cite{scs08}.  Nevertheless, since the lattice of partial frozen co-clones is rather well described within the Krom fragment (see \cite{Nordhz09}), this was an opportunity to popularize these co-clones, which are of independent interest.
\end{rem}

\section{Conclusion}
\label{sect:conclusion}

\begin{table*}[ht!]
\caption{Complexity of the Reasoning Problems (Completeness Results).}
\label{tab:1}
\begin{center}
\begin{tabular}{|l|l|l|l|l|}
\hline
& general case & Horn & Krom & Horn$\cap$Krom \\
\hline
\implication Satoh & 
$\PI 2$ & $\coNP$ &$\coNP$  &$\coNP$  \\
\mc Satoh & 
$\SIGMA 2$ & $\NP$ &$\NP$  &$\NP$  \\
\implication Dalal & 
$\THETA 2$ & $\THETA 2$ & $\THETA 2$ & $\THETA 2$\\
\mc Dalal & 
$\THETA 2$ & $\THETA 2$ & $\THETA 2$ & $\THETA 2$\\
$\leq$-{\sc Relevance}   & 
$\THETA 3$ & $\THETA 2$ & $\THETA 2$ & $\THETA 2$\\
\hline
\end{tabular}
\end{center}
\end{table*}


%
In this work we have investigated 
how the restriction of  propositional formulas to 
Krom affects the complexity of reasoning problems 
in the AI domain. 
Our results on belief revision and abduction are 
summarized in Table~\ref{tab:1}, where 
the complexity classifications 
for Krom and the combined Horn\,$\cap$\,Krom case 
refer to new results we have provided in the paper.
Having shown that 
the complexity of
problems involving 
cardinality minimality, like Dalal's revision operator or 
$\leq$-{\sc Relevance} in 
abduction, 
is often robust to such a restriction 
(even for formulas being 
Horn and Krom at the same time),
suggests that further classes within the Krom fragment should 
be considered to identify the exact tractability/intractability frontier.

The problem \cardmin seems to be the key for such investigations, thus
we initiated a deeper study by giving a complete classification of that
problem restricted to Krom within Schaefer's framework.
As a next step, 
we want to study the Krom form (and 
the yet more restricted fragments thereof) in the context of further
hard reasoning problems
and analyse which of these fragments suffice to yield tractability.


A complexity classification of $\cardminsat$ in the full propositional logic, which opens the door to many complexity results for reasoning problems  from the AI domain,  is under investigation.
 The lattice of partial frozen co-clones being insufficiently well-known, such a classification will require the use of
weak bases introduced in \cite{scs08}.

\section*{Acknowledgments}
This research has been supported by the Austrian Science Fund (FWF) through the
projects P25518, P25521and Y698, by the OeAD project FR 12/2013, and by the Campus France Amadeus project  29144UC.


\end{document}